\documentclass[11pt]{article}

\usepackage[T1]{fontenc}

\usepackage{epsfig}
\usepackage[outdir=./]{epstopdf}
\usepackage[english]{babel}
\usepackage{amsfonts,amsmath,enumerate}
\usepackage{amscd}
\usepackage{amsthm}
\usepackage{amsmath,amsfonts,amssymb,mathrsfs}
\usepackage{latexsym}
\usepackage{stmaryrd}
\usepackage{amssymb}
\usepackage{color}

\usepackage{booktabs}
\usepackage{multirow}
\usepackage{siunitx}

\usepackage{graphicx}
\usepackage{float}

\usepackage{rotating}
\usepackage{tikz}

\usepackage[numbers]{natbib}
\usepackage[titletoc]{appendix}

\newlength{\oldparindent}
\def\lbr{[\![}
\def\rbr{]\!]}

\definecolor{revised}{cmyk}{0,0,0,1}

\usepackage{fancyhdr}
\usepackage{hyperref}
\hypersetup{
	colorlinks = true,
	linkcolor = blue,
	anchorcolor = blue,
	citecolor = blue,
	filecolor = blue,
	urlcolor = blue
}

\usepackage{cleveref}
\crefname{equation}{equation}{equations}

\newtheorem{theorem}{Theorem}[section]

\newtheorem{lemma}[theorem]{Lemma}
\newtheorem{proposition}[theorem]{Proposition}
\newtheorem{example}[theorem]{Example}

\newtheorem{remark}[theorem]{Remark}

\newtheorem{assumption}[theorem]{Assumption}

\DeclareMathOperator*{\esssup}{ess\,sup}

\title{\Large \bf    \thanks{ } }
\author{\textsc{Caroline HILLAIRET\thanks{CREST, Ensae, Universit\'e Paris Saclay, 3 av Pierre Larousse
92245 Malakoff France. Email: caroline.hillairet@ensae.fr }
		\quad Cody HYNDMAN\thanks{Department of Mathematics and Statistics, Concordia University, 1455 Boulevard de Maisonneuve Ouest, Montr\'{e}al,
			Qu\'{e}bec, Canada H3G 1M8. Email: cody.hyndman@concordia.ca}
		\quad  Ying JIAO\thanks{ISFA, Universit\'{e} Lyon  1, 50 avenue Tony Garnier
			69007 Lyon, France.  Email:  ying.jiao@univ-lyon1.fr }
		\quad Renjie WANG \footnotemark[2]
	}\\
}

\date{\today}

\title{Trading against disorderly liquidation of a large position under asymmetric information and market impact}

\begin{document}
\maketitle
\numberwithin{equation}{section}
\begin{abstract}
We consider trading against a hedge fund or large trader that must liquidate a large position in a risky asset if the market price of the asset crosses a certain threshold. Liquidation occurs in a disorderly
 manner
and  negatively impacts the market price of the  asset.  %
We consider the perspective of small investors whose trades do not induce market impact and who possess different levels of information about the liquidation trigger mechanism and the market impact. We classify these market participants into three types: fully informed, partially informed and uninformed investors.   We consider the portfolio optimization problems and compare the optimal trading and wealth processes for the three classes of investors theoretically and by numerical illustrations.
\end{abstract}

\vspace*{0.5cm}

\newpage

\section{Introduction}
There is a large literature on insider trading, asymmetric information, and market manipulation trading strategies including seminal works by \citet{kyle1985continuous}, \citet{back1992insider}, \citet{jarrow1992market,jarrow1994derivative}, and \citet{allen1992stock}. These works generally assume that an insider is attempting to influence a price by, or profit from, the release of, potentially false, information known to the insider. These studies also generally break market participants into noise traders, standard informational traders, and informed traders. The existence of arbitrage strategies, price equilibrium, or specific market manipulation strategies are the primary concerns of these early works. Other papers dealing with insider information which quantify the value of insider information through the maximization of agents wealth or utility include \citet{MR98d:90034}, \citet{elliott1999incomplete},  \citet{amendinger1998additional},  and  \citet{amendinger2003monetary}.

More  recently  liquidity  modeling  has  become  an intense area of study.  Market micro-structure and limit order books present one approach to modelling liquidity based on trading mechanisms.  Models that specify the price impact of trades as exogenously determined and depending on the size of a transactions constitute another strand of the literature.  Both approaches treat problems associated with the fact that trading large positions impacts market prices.  A good overview of liquidity models can be found in \citet{gokay2011liquidity}.  The modeling of market micro-structure and the optimal liquidation of large positions has also been studied extensively and an overview of these topics can be found in \citet{abergel2012market}. To the best of our knowledge, among works dealing with asymmetric information, only few papers concern the market impact of liquidation risk.  In particular, \citet{ankirchner2012optimal} studies optimal liquidation problems of an insider.

In contrast to the existing literature we are concerned with disorderly, rather than optimal, liquidation and the point of view of market participants other than the large trader or hedge fund liquidating the position. In particular, we are interested in the following question: is it possible for a market participant to profit from the knowledge that another market participant, with large positions in a stock or derivative, will be forced to liquidate some or all of its position if the price crosses a certain threshold?  There is ample evidence from financial markets concerning the importance of liquidity risks. For example, consider a hedge fund with a large position in natural gas futures contracts, such as Amaranth Advisors LLC in 2006, and macro-economic or weather events contribute to an unexpected adverse change in the price.  In this case the fund may be forced to unwind its positions in a disorderly fashion, which would have a further market impact on the price. Other examples include the case of Long Term Capital Management L.P. (LTCM) in 1998 and numerous firms during the financial crisis of 2007-2008.

 We assume that liquidation occurs immediately when the market price hits the liquidation trigger level and has a temporary impact on the asset price, whereby the market price is depressed away from the fundamental value, and gradually dissipates. We model the temporary market impact by a function with parameters that control the impact speed and magnitude. Other market participants may have different levels of information about the liquidation trigger mechanism and the liquidation impact. We aim to find the optimal trading strategy that maximizes an investor's terminal utility of wealth under different types of information that are accessible to particular market participants. In the standard information case an uninformed market participant is not aware of the liquidation trigger mechanism. They believe and act, erroneously when liquidation occurs, as if the market price is equal to the fundamental asset price. In the partial information case an insider or informed market participant knows the level at which the hedge fund will be forced to liquidate the position but does not have information about the liquidation volume which determines the price impact. In the full information case the insider has complete information about the liquidation threshold and the price impact. Certain market participants may have access to this type of information owing to their position, counter-party status, technology, or knowledge of the market.  The fully informed investor's perfect information represents one extreme which may be unobtainable in practice.  However, we shall show numerically in the power-utility case that the optimal strategy for the partially informed investor is quite close to that of the fully informed investor.

The remainder of the paper is organized as follows. Section \ref{section:model_setting} sets up the framework of our model. Section \ref{sect:fully_informed} solves the portfolio optimization problem for a fully informed investor and gives the explicit expression of  the optimal expected utility in case of log utility. Sections \ref{sect:partially_informed} and \ref{sect:uninformed} explore the optimization problems for uninformed and  partially informed investors, respectively.  Section \ref{sect:numerical_result} presents some numerical results. Section \ref{sect:conclusion} concludes and an appendix contains technical results and proofs.

\section{The Market Model}\label{section:model_setting}
\subsection{Asset price and liquidation impact}
Fix a probability space $(\Omega, \mathfrak{A},\mathbb{P})$ equipped with a  reference  filtration $\mathbb{F}=(\mathcal{F}_{t})_{t\geq 0}$ satisfying the usual conditions, with
$(W_{t},t\geq 0)$ an  $(\mathbb{F}, \mathbb{P})$-Brownian motion. Let $T> 0$ be a finite horizon time.
In our model, we assume that market participants may invest in a riskless asset and a risky asset. Without loss of generality we suppose that the interest rate of the riskless asset is zero.  The fundamental value of the risky asset is modelled by a Black-Scholes diffusion:
\begin{equation}
dS_{t}=S_{t}(\mu  dt+\sigma dW_{t}), \quad 0 \leq t \leq T,
\label{eq:risky_asset_fundamental_value}
\end{equation}
where $\mu$ and $\sigma$ are supposed to be constants, and $\sigma>0$.

We consider a hedge fund which holds a large long position in the risky asset over the investment horizon $[0,T]$. In normal
circumstances, this position could be held until time $T$. However, according to risk management policies, exchange rules, or regulatory requirements,
 the long position must be liquidated in certain circumstances. In this paper, we assume that the liquidation will be triggered  when the market price of the risky asset passes below a pre-determined level. Before liquidation, the market price, denoted by $S^{M}$, is equal to the fundamental value $S$. So the liquidation time $\tau$ is defined as the first passage time of a fixed constant threshold $ \alpha S_{0}$ where $\alpha \in (0,1)$, by the market price process $S^{M}$, i.e.,
\begin{equation}
\tau:=\inf\{t\geq 0, S_{t}^{M}\leq \alpha S_{0} \}=\inf\{t\geq 0, S_{t}\leq \alpha S_{0} \}%
\label{eq:liquidation_time_definition}
\end{equation} %
with the convention $\inf\varnothing =\infty$.  We note that $\tau$ is an $\mathbb{F}$-stopping time. In the simplest case the scenario described corresponds to a margin call that cannot be covered resulting in the liquidation, in full or in part, of the position.

The market price of the risky asset will be influenced by liquidation. Since the number of shares of the risky asset to be sold is very large in comparison to the average volume traded in a short time period, immediate liquidation would have a temporary impact on the market price which would be driven down away from the fundamental price after
liquidation. We denote by  $S^{I}_{t}(u)$  the market price of the risky asset at time $t$ after the liquidation time $\tau=u$. Suppose that it is given as
\begin{equation}
S_{t}^{I}(u)=g(t-u;\Theta,K)S_{t},\quad  u\leq t\leq T.
\label{eq:impacted_risky_asset_price}
\end{equation}
where $g$ is an impact function and  $\Theta$  and $K$ are parameters which will be made precise later.  We note that the mathematical characterization of market impact is a very complicated problem,  and we refer the interested reader to  \citet{kissell2003optimal} for details.  In this paper, inspired by \citet{Li2014}, we characterize the temporary influence of liquidation on market by the impact function $g$ of the form
\begin{equation}
g(t;\Theta,K)=1-\frac{Kt}{\Theta} e^{1-\frac{t}{\Theta}}
\label{eq:impact_function}
\end{equation}
where $\Theta$ and $K$ are positive parameters with  $\Theta$ controlling the speed of the market impact and $K$ representing the magnitude of the market impact. In particular, we assume that $\Theta$ is a positive random variable and $K$ is a random variable valued in $[0,1]$, both of which are independent of $\mathbb{F}$ and with joint probability density function $\varphi(\cdot,\cdot)$, i.e. $\mathbb{P}(\Theta \in d\theta, K \in dk)=\varphi(\theta,k)d\theta dk$.

Figure \ref{fig:impact_function} illustrates the impact function (\ref{eq:impact_function}) with $K=0.1$ and two different realized values of $\Theta$. Clearly the shape of the impact function with $\Theta=0.05$ is steeper than with $\Theta=0.1$.  We note that for each fixed scenario $\omega$, the function $g$ attains its minimum value $1-K(\omega)$ at $t=\Theta(\omega)$.
Also, we observe that the function $g$ first declines from $1$ and then rises back and converges to $1$, which characterizes the market impact of liquidation with time evolution. For realized values $K=0.1$ and $\Theta=0.1$ it would take 0.1 year, which is approximately 25 trading days, for the asset price to reach the minimum value $(1-K)*S_{0}$  after liquidation occurs.  The market impact in the first trading day after liquidation is $1 - g(\frac{1}{250};0.1,0.1)\approx 1\%$. Therefore, the parameter $\Theta$ needs to be small to more accurately reflect the impact of disorderly liquidation. In Section \ref{sect:numerical_result} we present some numerical results which use a rather large $\Theta$ that guarantees better accuracy of the numerical results, but these could be improved by applying other numerical techniques for smaller values of $\Theta$.

\begin{figure}[H]
\centering
\begin{minipage}{.45\textwidth}
  \centering
  \includegraphics[width=\linewidth]{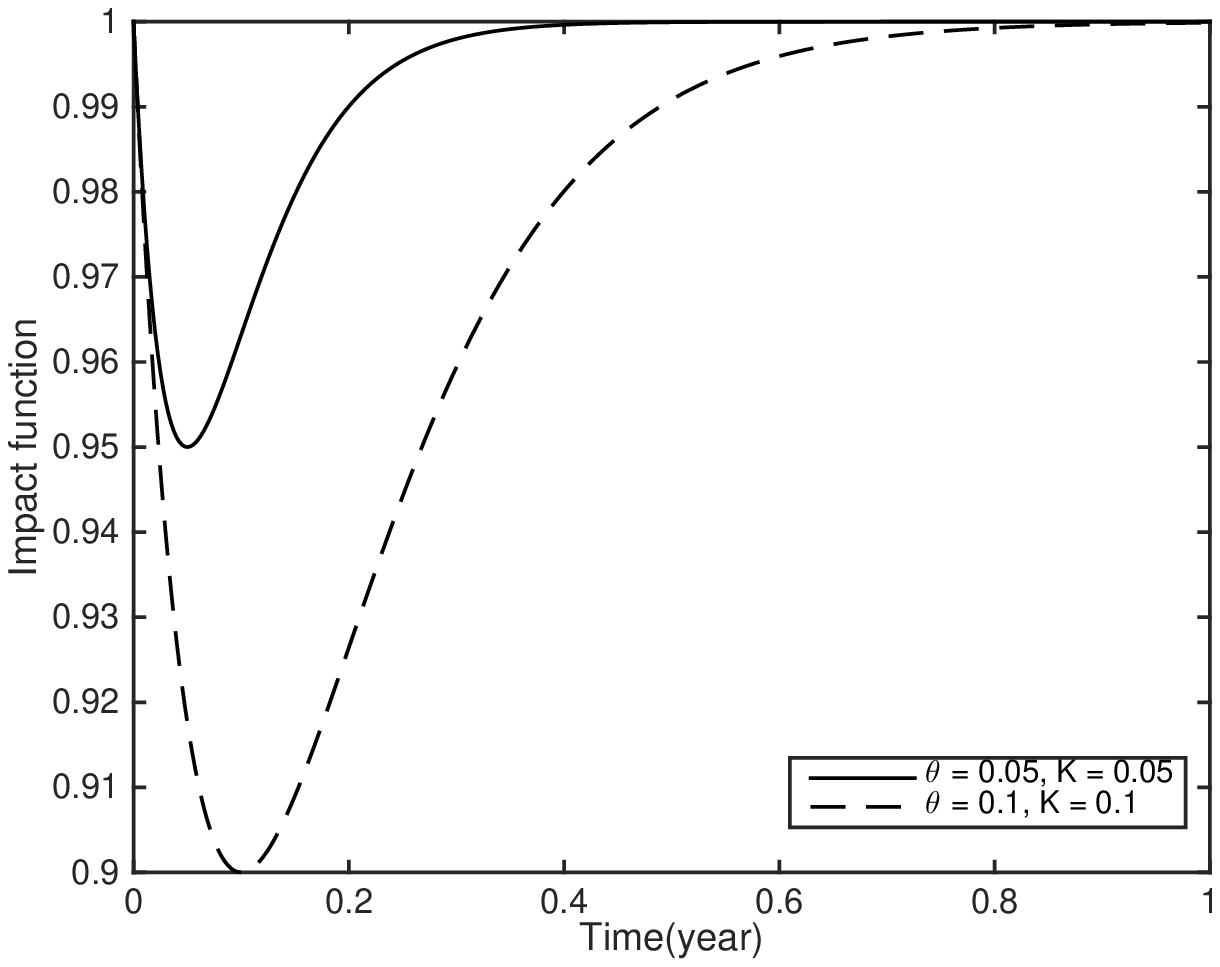}
 \caption{Impact function with 2 parameters}
\label{fig:impact_function}
\end{minipage}
\begin{minipage}{.45\textwidth}
  \centering
  \includegraphics[width=\linewidth]{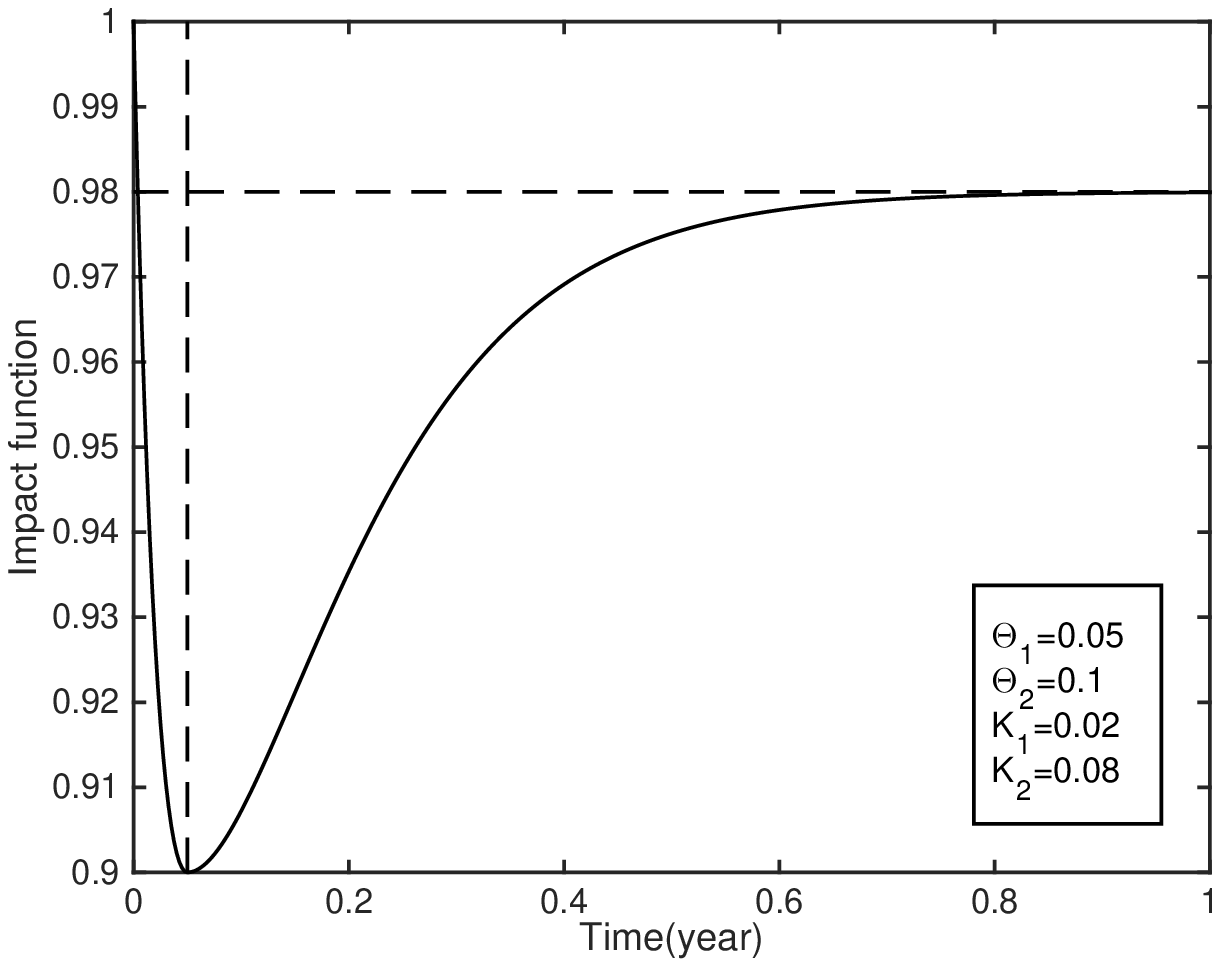}
  \caption{Impact function with 4 parameters}
\label{fig:impact_function4}
\end{minipage}
\end{figure}

\begin{remark}
It is natural to consider a jump effect for the price impact of liquidation.  In our model, by (\ref{eq:impacted_risky_asset_price}), the price before and just after liquidation satisfies the relation $S_{t}^{I}(t)=S_{t}$. However, we can approximate downward jumps of asset prices after liquidation by choosing small values of $\Theta$ in the smooth function $g$.  Further, our model allows us to consider the situation that liquidation by the large trader may have no long-term informational content.  The temporary impact on the market price decays as liquidity providers return to the market and other market participants realize that there may be no information about the fundamental value of the risky asset conveyed by the hedge fund's disorderly liquidation.

A possible extension is to consider a modified impact function with additional parameters and flexibility. For example, let
\begin{equation}
g(t;\Theta_{1},\Theta_{2},K_{1},K_{2})=
\begin{cases} 1-\frac{(K_{1}+K_{2})t}{\Theta_{1}} e^{1-\frac{t}{\Theta_{1}}} & 0 \leq t < \Theta_{1}, \\
1-K_{1}-\frac{K_{2}(t+\Theta_{2}-\Theta_{1})}{\Theta_{2}} e^{1-\frac{t+\Theta_{2}-\Theta_{1}}{\Theta_{2}}} &  \Theta_{1} \leq t.
\end{cases} \label{eq:impact_function4par}
\end{equation}
The  impact function given by (\ref{eq:impact_function4par}) incorporates both permanent and temporary market impacts  with $K_{1}$ and $K_{2}$ controlling the magnitude of permanent and temporary market impacts respectively. The  parameters $\Theta_{1}$ and $\Theta_{2}$ determine both the deviation and reversal speed (see Figure \ref{fig:impact_function4}). Moreover, at long-term time scale, the impact function can come back to a different level other than $1$.
For simplicity, we will use the impact function given by (\ref{eq:impact_function}) in this paper and suppose the parameters $\Theta$ and $K$ to be random variables.
\end{remark}

Considering the market price of the asset to be equal to the fundamental value before the liquidation time $\tau$  and to be the impacted asset price after liquidation, we have that the market price  is given as
$$S^{M}_{t}=1_{\{0\leq t <\tau\wedge T \}} S_{t}+1_{\{\tau\wedge T \leq t \leq T\}}S^{I}_{t}(\tau)$$
where $S_{t}$ and $S^{I}_{t}(\tau)$ are given by (\ref{eq:risky_asset_fundamental_value}) and (\ref{eq:impacted_risky_asset_price}) respectively. Moreover, for any $u\geq 0$, the dynamics of the process $S^{I}_{t}(u)$ satisfies the SDE
\begin{equation}
dS_{t}^{I}(u)=S_{t}^{I}(u)\left( \mu^{I}_{t}(u,\Theta,K)dt+\sigma dW_{t} \right), \quad u\leq t\leq T
\label{eq:after_liquidation_risky_asset_dynamics}
\end{equation}
where
\begin{equation}
\mu^{I}_{t}(u,\Theta,K)=\frac{g^\prime (t-u;\Theta,K)}{g(t-u;\Theta,K)}+\mu.
\label{eq:drift_term_market_price}
\end{equation}

\begin{remark}\label{remark:brownian_motion}
The process $(S_{t}^{I}(u),t\geq u)$ is adapted with respect to the filtration $\mathbb F\vee\sigma(\Theta, K)$ which is the initial enlargement of $\mathbb F$ by the random variables $(\Theta,K)$. As we suppose $\sigma(\Theta,K)$ is independent of $\mathbb{F}_{\infty}$,  the $(\mathbb{F},\mathbb{P})$-Brownian motion $W$ is also a $(\mathbb{F}\vee \sigma(\Theta,K) ,\mathbb{P})$-Brownian motion (see e.g. \citet[Section 5.9]{jeanblanc2009mathematical}.)
\end{remark}
\noindent Thus the market price process of the risky asset, denoted as $S^{M}=(S_{t}^{M},t\geq 0)$, satisfies the SDE
\begin{equation}
dS_{t}^{M}=S^{M}_{t}\left( \mu^{M}_{t}(\Theta,K) dt+\sigma dW_{t} \right)
\label{eq:global_risky_asset_dynamics}
\end{equation}
where
\begin{equation}
\mu^M_{t}(\Theta,K)=1_{\{ 0\leq t <\tau\wedge T   \}}\mu +1_{\{\tau\wedge T \leq t \leq T\}} \mu^{I}_{t}(\tau,\Theta,K).
\label{eq:mu_global_expression}
\end{equation}
We note that the market price admits a regime change at the liquidation time $\tau$, in particular on the drift term. We give an illustrative example as below.

\begin{example}
Suppose that the fundamental value process (\ref{eq:risky_asset_fundamental_value}) is given by the Black-Scholes model with parameters
$S^{M}_{0}=80, \mu=0.07, \sigma=0.2, \alpha=0.9, \Theta=0.1, K=0.1$. Figure \ref{fig:drift_term} shows that liquidation triggers a downward jump of the drift term. Afterward the drift term first rises quickly and then declines gradually back to the original drift term.  Correspondingly, Figure \ref{fig:asset_price} shows the sample market price processes of the asset subject to liquidation impact compared with the fundamental value process.
	
\begin{figure}
\centering
\begin{minipage}{.45\textwidth}
  \centering
  \includegraphics[width=\linewidth]{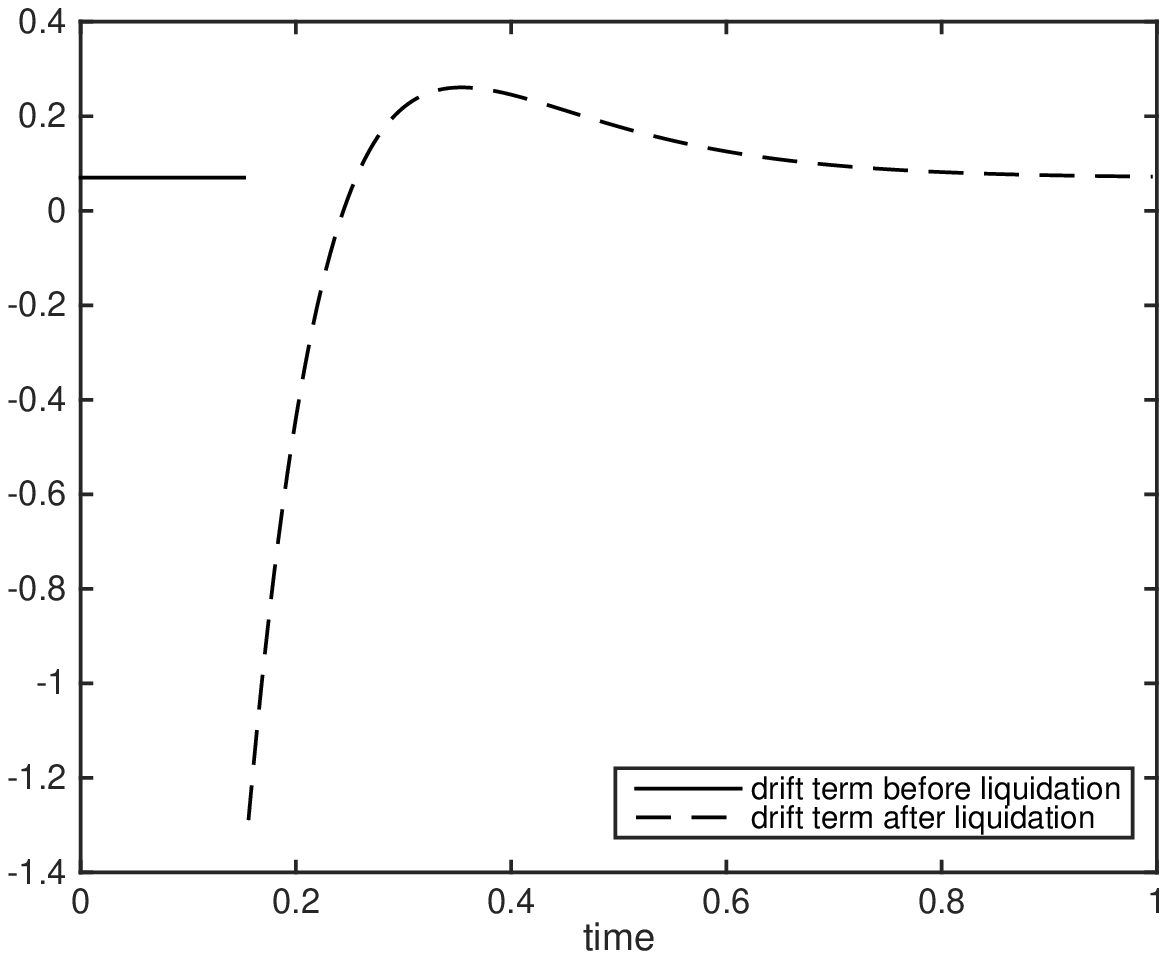}
 \caption{Drift $\mu$ before and after liquidation}
\label{fig:drift_term}
\end{minipage}
\begin{minipage}{.45\textwidth}
  \centering
  \includegraphics[width=\linewidth]{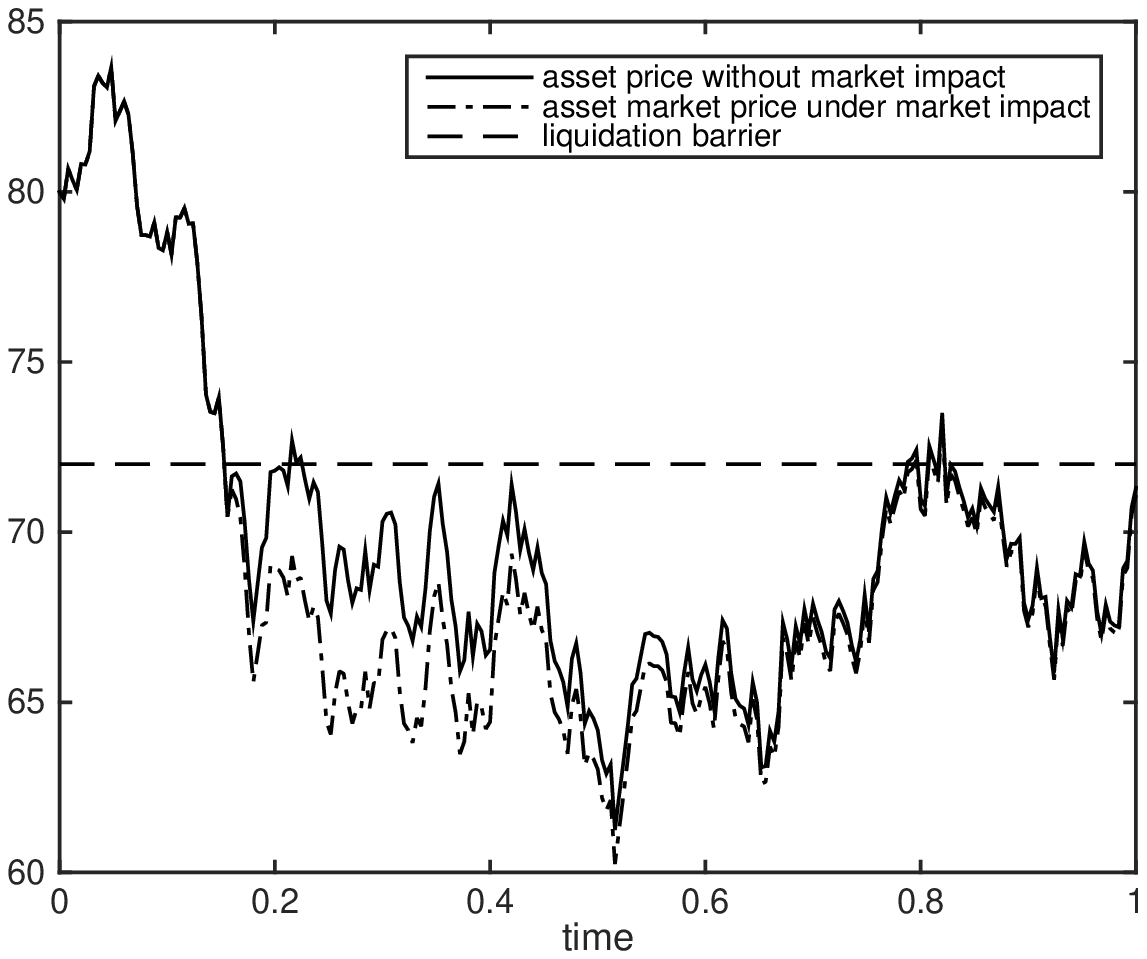}
  \caption{Corresponding asset price $S^M$}
\label{fig:asset_price}
\end{minipage}
\end{figure}
\end{example}

\subsection{The optimal investment problem}
Our objective is to consider the optimal investment problem from the  perspective of investors who trade in the market for the risky asset subject to price impact from disorderly liquidation of the hedge fund's position.  For simplicity we assume these agents may trade in the market for the risky asset without transaction costs.  We consider fully informed investors, partially informed investors and uninformed investors. We suppose that all investors have access to the  market price of the risky asset $S^M$ but their knowledge of the liquidation and price impact are different.  We further assume that  all  the investors know the values of the parameters $\mu$ and $\sigma$.

Fully informed investors observe the market price and are assumed to have complete knowledge of the mechanism of liquidation and the price impact function. Hence they know, in mathematical terms, the liquidation trigger level $\alpha$, the impact function $g$, and the values of the random variables  $\Theta$ and $K$ when liquidation occurs. Therefore, fully informed investors have complete knowledge of the dynamics of the market price process, together with the information of the price impact.  %

Partially informed investors are also able to observe the market price and know the liquidation trigger level $\alpha$, therefore, the liquidation time $\tau$ is also observable for them.  However, partially informed investors
do not have complete information about the price impact function.  We suppose the partially informed investors know the functional form of the price impact function $g$.  However, we assume the partially informed investors only know the distributions of $\Theta$ and $K$ but not the realized value that is necessary to have full knowledge of the price impact of liquidation. %

Uninformed investors are not aware of the liquidation trigger mechanism. They erroneously believe the market price process  follows the  Black-Scholes  dynamics (\ref{eq:risky_asset_fundamental_value})  without price impact. Considering uninformed investors allows us to quantify the value of information about the liquidation barrier and price impact, compared to a Merton-type investor.

We denote by $\mathbb F^{S}=(\mathcal{F}^{S}_{t})_{t\geq 0}$ the natural filtration generated by the market price process $S^{M}$. Since the market price coincides with the fundamental value process $S$ before liquidation, the liquidation time $\tau$, which is an $\mathbb{F}$-stopping time, is also an $\mathbb{F}^{S}$-stopping time. We summarize the knowledge of the various investors in the following assumption.

\begin{assumption}\label{assumption:information_investors}
 All investors observe the market price of the risky asset and know the values of the parameters $\mu$ and $\sigma$.  In addition certain market participants possess additional  information:
\begin{list}{(\roman{enumi})}{\usecounter{enumi}}
\item The observable information for fully informed investors is modeled by the filtration $$\mathcal{G}^{(2)}_{t}=\mathcal{F}^{S}_{t} \vee \sigma(\Theta,K) = \mathcal{F}_{t} \vee \sigma(\Theta,K),$$ they further know the liquidation barrier $\alpha$, as well as the form of the impact function $g$.
\item The observable information for partially informed investors is modeled by the filtration $$\mathcal{G}^{(1)}_{t}=\mathcal{F}^{S}_{t}, $$  they further  know the liquidation barrier $\alpha$, the form of the impact function $g$, and  the distribution of $(\Theta, K)$.
\item To compare with the above two types of insiders, we consider uninformed investors whose observable information is modeled by  the filtration $\mathcal{G}^{(0)}_{t}=\mathcal{F}^{S}_{t} $. They have no information  about the liquidation mechanism. Further, they do not update their knowledge of the drift process after $\tau$ and act as Merton-type investors,  erroneously considering  Black-Scholes dynamics with constant $\mu$ over the entire period $[0,T]$.
\end{list}
\end{assumption}

\begin{remark}

The common information to three types investors are represented by the "public" filtration $\mathbb{F}^{S}$ since the market price of the risky asset is observable to all investors. Assumption \ref{assumption:information_investors} implies that partially informed investors know the the law of $\mu^{M}_{t}(\Theta,K)$. This is similar to the weak information case of \citet{baudoin2003modeling}.

The essential differences among these three types of investors lie in  their knowledge on the drift term $\mu^{M}(\Theta,K)$ defined in (\ref{eq:drift_term_market_price}). Fully informed investors are able to completely observe the drift term.  Partially informed investors partially observe the drift term, corresponding to the case of partial observations considered by \citet{karatzas1998bayesian}. Partially informed investors may obtain an estimate of the drift term which is adapted to their observation process using filtering theory. Uninformed investors do not have any information about the liquidation mechanism and market impact which causes them to erroneously specify the drift term as $\mu$. That is, uninformed investors believe that the market prices follow the Black-Scholes dynamics (\ref{eq:risky_asset_fundamental_value}).  If the uninformed investor treated the drift of (\ref{eq:risky_asset_fundamental_value}) as an unobservable process he could perhaps apply filtering theory to improve his investment decisions even without knowing anything about the liquidation mechanism or market impact function.  However, in this paper we shall only consider the case of Assumption~\ref{assumption:information_investors}, that is of  uninformed investors who estimate the drift at the beginning of the  period and do not update it, since from their own view point no liquidation event happened during the period $[0,T]$.
The uninformed investors are mainly considered as a benchmark for comparison with the Merton model.
	
\end{remark}

We shall study the portfolio optimization problem for three types of investors in the remainder of this paper under logarithmic and power utility.

\section{Fully informed investors}\label{sect:fully_informed}

Fully informed investors choose their trading strategy to adjust the portfolio of assets according to their information accessibility.
 As discussed in Section \ref{section:model_setting} fully informed investors know the realized values of the random variables $\Theta$ and $K$.
The investment strategy  is characterized by a $\mathbb{G}^{(2)}$-predictable process $\pi^{(2)}$ which represents the proportion of wealth invested in the risky asset.  The admissible strategy set $\mathcal{A}^{(2)}$ is a collection of $\pi^{(2)}$
such that, for any $(\theta,k)\in (0,+\infty)\times (0,1)$,
\begin{equation}
\int_{0}^{T} | \pi^{(2)}_{t}\mu^{M}_{t}(\theta,k) | dt +\int_{0}^{T}|\pi^{(2)}_{t}\sigma|^{2}dt< \infty.
\label{eq:admissible_strategy_integration_condition}
\end{equation}
The risk aversion of the investors is modeled by classic utility functions $U$ defined on $(0,\infty)$ that are strictly increasing,  strictly concave, with continuous derivative $U'(x)$ on $(0,\infty)$, and satisfying
$$\lim_{x\rightarrow 0^{+}}U'(x)= +\infty \quad \quad \lim_{x\rightarrow\infty}U'(x)=0.$$
We define the $\mathbb{G}^{(2)}$-martingale measure  $\mathbb{Q}$ by the likelihood process
\begin{equation}
L_{t}:=\left.\frac{d\mathbb{Q}}{d\mathbb{P}}\right|_{\mathcal{G}^{(2)}_t}=\exp\left\{ -\int_{0}^{t}\frac{\mu^M_{v}(\Theta,K)}{\sigma}dW_{v}-\int_{0}^{t}\frac{\left(\mu^M_{v}(\Theta,K)\right)^{2}}{2\sigma^{2}}dv \right\}.
\label{eq:likelihood_process_fully_informed}
\end{equation}
As mentioned in  Remark \ref{remark:brownian_motion}, $W$ is a $(\mathbb{G}^{(2)},\mathbb{P})$-Brownian motion. By Girsanov's theorem, the process $W^{\mathbb{Q}}$ defined as
\begin{equation}
W^{\mathbb{Q}}_{t} = W_{t} + \int_{0}^{t}\frac{\mu^M_{v}(\Theta,K)}{\sigma}dv
\label{eq:girsanov_brownian_change_fully_informed}
\end{equation}
is an $(\mathbb{G}^{(2)}, \mathbb{Q})$-Brownian motion and the dynamics of the asset price $S^M$ under $\mathbb{Q}$ may be written as
\begin{equation}
dS_{t}^{M}=S^{M}_{t}\sigma dW_{t}^{\mathbb{Q}}.
\end{equation}

By taking  a strategy $\pi^{2}\in  \mathcal{A}^{(2)}$, the wealth process with initial endowment $X_{0} \in \mathcal{G}^{(2)}_{0}$
 evolves as
\begin{equation}
dX_t^{(2)}=X_t^{(2)}\pi_t^{(2)}(\mu^M_t(\Theta,K) dt+\sigma dW_t), \quad \quad 0 \leq t \leq T
\label{eq:wealth_fully_informed}
\end{equation}
that is
\begin{equation}
X^{(2)}_{t} = X_{0}+\int_{0}^{T}  \pi^{(2)}_{v}\sigma S^{M}_{v}dW^{\mathbb{Q}}_{v}.
\label{eq:wealth_under_Q_fully_informed}
\end{equation}
{\color{revised}  The fully informed investor's objective is to maximize her expected utility of terminal wealth
\begin{equation}
V^{(2)}_{0} := \sup_{\pi^{(2)}\in \mathcal{A}^{(2)}} \mathbb{E} \left[U\left(X^{(2)}_{T}\right)\right]
\label{eq:optimization_objective_fully_informed}
\end{equation}
or
\begin{equation}
V^{(2)}_{0}(\Theta,K):=  \esssup_{\pi^{(2)}(\Theta,K)\in \mathcal{A}^{(2)}} \mathbb{E} \left[U\left(X^{(2)}_{T}\right)|\mathcal{G}^{(2)}_{0}\right]
\label{eq:optimization_objective_fully_informed_conditional}
\end{equation}
where $\mathcal{G}^{(2)}_{0} = \sigma(\Theta,K)$.
The link between the optimization problems (\ref{eq:optimization_objective_fully_informed}) and 	
(\ref{eq:optimization_objective_fully_informed_conditional}) is
given by \citet{amendinger2003monetary}; if the supremum in (\ref{eq:optimization_objective_fully_informed_conditional})
is attained by some strategy in $\mathcal{A}^{(2)}$, then the $\omega$-wise optimum is also a solution to
(\ref{eq:optimization_objective_fully_informed}). }

 As $(\Theta,K)$ is independent\footnote{This assumption can be relaxed into a density Jacod hypothesis, using then the result of  \citet[Proposition 4.6]{amendinger2000martingale} for a martingale representation theorem.} of $\mathbb F$, a martingale representation theorem holds for $(\mathbb{G}^{(2)}, \mathbb{Q})$-local martingale, thus we adopt the standard "martingale approach" (see \citet{karatzas1998methods}) to solve the utility optimization problem (\ref{eq:optimization_objective_fully_informed_conditional}).
We may consider the following static optimization problem
\begin{equation}
\sup_{X^{(2)}_{T}\in \mathcal{V}} \mathbb{E} \left[U\left(X^{(2)}_{T}\right) |\mathcal{G}^{(2)}_{0}\right]
\label{eq:static_optimization_fully_informed}
\end{equation}
where $\mathcal{V} = \left\{ X^{(2)}_{T} \left| X^{(2)}_{T} = X_0 + \int_{0}^{T} \pi^{(2)}_v\sigma S^{M}_vdW^{\mathbb{Q}}_{v} \right. ,  ~\pi^{(2)}\in \mathcal{A}^{(2)}  \right\}$.
The optimization problem (\ref{eq:static_optimization_fully_informed}) can be solved by using  the method of Lagrange multipliers (see \citet[Proposition 4.5]{amendinger2003monetary}). The optimal terminal wealth $\hat X^{(2)}_{T}$ is given by
\begin{equation}
\hat X^{(2)}_{T}=I(\Lambda L_{T}),
\label{eq:optimal_terminal_wealth_fully_informed}
\end{equation}
where $I=(U^\prime)^{-1}$ and the $\mathcal{G}^{(2)}_{0}$-measurable  random variable $\Lambda$ is determined by
\begin{equation}
\mathbb{E}^{\mathbb{Q}}\left[\left. I(\Lambda L_{T})\right |\mathcal{G}^{(2)}_{0}\right]=X_{0}.
\label{eq:lambda_solution_fully_informed}
\end{equation}
In order to find the optimal strategy $\hat\pi^{(2)}$ one should provide the dynamics of the optimal wealth process
\begin{equation}
\hat X^{(2)}_{t}=\mathbb{E}^{\mathbb{Q}}\left[\hat X^{(2)}_{T}|\mathcal{G}^{(2)}_{t}\right].
\label{eq:optimal_wealth_process_fully_informed}
\end{equation}
Since $(X^{(2)}_{t})_{t\in[0,T]}$ is a $(\mathbb{G}^{(2)},\mathbb{Q})$-martingale, there exists a $\mathbb{G}^{(2)}$-adapted process $J$ such that
\begin{equation}
\hat X^{(2)}_{t}=\mathbb{E}^{\mathbb{Q}}[\left. \hat{X}^{(2)}_{T} \right| \mathcal{G}^{(2)}_{0}]+\int_{0}^{t}J_{v}dW^{\mathbb{Q}}_{v}.
\label{eq:wealth_martingale_representation_fully_informed}
\end{equation}
Substituting (\ref{eq:lambda_solution_fully_informed}) into (\ref{eq:wealth_martingale_representation_fully_informed}) we have
\begin{equation}
\hat X^{(2)}_{t}=X_{0}+\int_{0}^{t}J_{v}dW^{\mathbb{Q}}_{v}.
\label{eq:wealth_martingale_representation_fully_informed_with_inital_x0}
\end{equation}

Comparing (\ref{eq:wealth_under_Q_fully_informed}) with (\ref{eq:wealth_martingale_representation_fully_informed_with_inital_x0}), we obtain the optimal strategy
\begin{equation}
\hat\pi^{(2)}_{t}=\frac{J_{t}}{\sigma \hat X^{(2)}_{t}}.
\end{equation}
Notice that the optimal strategy $(\hat\pi^{(2)}_{t})_{t\in[0,T]}$ involves the process $J$ which is implicitly determined by the martingale representation as in (\ref{eq:wealth_martingale_representation_fully_informed}). To obtain an explicit expression for the optimal strategy, we will consider power and logarithmic utilities in the following subsections.

\subsection{Power utility} \label{sect:fully_power_utility}
We first consider the power utility  $U(x)=\frac{x^p}{p}$, $0<p<1$.  Using the fact that $I(x) = x^{p-1}$ and by (\ref{eq:optimal_terminal_wealth_fully_informed})-(\ref{eq:lambda_solution_fully_informed}) we obtain the optimal terminal wealth
\begin{equation}
\hat X^{(2)}_{T}=\frac{X_{0}}{\mathbb{E}\left[\left(L_{T}\right)^{\frac{p}{p-1}}\right]}\left(L_{T}\right)^{\frac{1}{p-1}}
\label{eq:optimal_terminal_wealth_power_fully_informed}
\end{equation}
where $L_{T}$ is given by (\ref{eq:likelihood_process_fully_informed}).
The following proposition gives then the  optimal expected utility as well as the optimal strategy:

\begin{proposition}
 For  power utility  $U(x)=\frac{x^p}{p}$, $0<p<1$, the optimal expected utility is
 \begin{equation}
V^{(2)}_{0}(\Theta,K)=\frac{(X_{0})^{p}}{p}\left(\mathbb{E}\left[\left(L_{T}\right)^{\frac{p}{p-1}}\right]\right)^{1-p}
\label{eq:optimal_utility_power_fully_informed}
\end{equation}
  and  the optimal strategy is given by
\begin{equation}
\hat\pi_t^{(2)}=1_{\{ 0 \leq t < \tau \wedge T  \}}\hat \pi^{(2,b)}_{t}+1_{\{  \tau \wedge T \leq t \leq  T \}}\hat \pi_{t}^{(2,a)}, \quad t \in[0,T]
\label{eq:optimal_strategy_power_fully_informed_decomposed}
\end{equation}
where
\begin{align}
\hat \pi^{(2,b)}_{t} & = \frac{\mu}{(1-p)\sigma^2}+\frac{ Z^{H}_{t}}{\sigma H_{t}}, \quad  t \in  \lbr 0, \tau\wedge T\lbr ,\label{eq:optimal_strategy_power_fully_informed_before_liquidation}\\
\hat \pi^{(2,a)}_{t} & = \frac{\mu_{t}^{I}(\tau , \Theta,K)}{(1-p)\sigma^2}, \quad    t \in  \lbr \tau\wedge T, T \rbr
\label{eq:optimal_strategy_power_fully_informed_after_liquidation}
\end{align}
with $(H,Z^{H})$ satisfying the following linear BSDE
\begin{equation}
H_{t} = 1+\int_{t}^{T}\left( \frac{p\left(\mu^M_v(\Theta,K)\right)^2}{2(1-p)^{2}\sigma^{2}}H_{v}+\frac{p\mu^M_v(\Theta,K)}{(1-p)\sigma} Z^{H}_{v} \right)dv-\int_{t}^{T}Z^{H}_{v}dW_{v}.
\label{eq:H_BSDE_full_before_liquidation}
\end{equation}
\end{proposition}

\begin{proof}

Following \citet{bjork2010optimal}  we find the explicit expression for the optimal strategy, by computing the dynamics of the optimal wealth process. Applying the abstract Bayes' formula to (\ref{eq:optimal_wealth_process_fully_informed}), we obtain
\begin{align}
\hat X^{(2)}_{t} =& \mathbb{E}^{\mathbb{Q}}\left[\hat X^{(2)}_{T}|\mathcal{G}^{(2)}_{t}\right] \nonumber\\
=& \frac{1}{L_{t}} \mathbb{E}\left[\hat X^{(2)}_{T}L_{T}|\mathcal{G}^{(2)}_{t}\right]. \label{eq:optimal_termimal_wealth_bayes_power_fully_informed}
\end{align}
Substituting (\ref{eq:likelihood_process_fully_informed}) and (\ref{eq:optimal_terminal_wealth_power_fully_informed}) into (\ref{eq:optimal_termimal_wealth_bayes_power_fully_informed}) we have
\begin{align}
\hat X^{(2)}_{t} =& \frac{X_{0}}{\mathbb{E}[(L_{T})^{\frac{p}{p-1}}]L_{t}} \mathbb{E}\left[ (L_{T})^{\frac{p}{p-1}}|\mathcal{G}^{(2)}_{t}\right] \nonumber \\
=& \frac{X_{0}}{\mathbb{E}[(L_{T})^{\frac{p}{p-1}}]L_{t}}\mathbb{E}\left[ \exp\left\{ \int_{0}^{T}\frac{p\mu^M_{v}(\Theta,K)}{(1-p)\sigma}dW_{v}+\int_{0}^{T}\frac{p\left(\mu^{M}_{v}(\Theta,K)\right)^{2}}{2(1-p)\sigma^{2}}dv \right\}|\mathcal{G}^{(2)}_{t}\right] \nonumber \\
=&  \frac{X_{0}(L_{t})^{\frac{1}{p-1}}}{\mathbb{E}[(L_{T})^{\frac{p}{p-1}}]}\mathbb{E}\left[ \exp\left\{ \int_{t}^{T}\frac{p\mu^{M}_{v}(\Theta,K)}{(1-p)\sigma}dW_{v}+\int_{t}^{T}\frac{p\left(\mu^{M}_{v}(\Theta,K)\right)^{2}}{2(1-p)\sigma^{2}}dv \right\}|\mathcal{G}^{(2)}_{t}\right]
\end{align}
Defining
\begin{equation}
H_{t}:=\mathbb{E}[\exp\left\{ \int_{t}^{T} \frac{p\mu_{v}^{M}(\Theta,K)}{(1-p)\sigma}dW_{v} +\int_{t}^{T} \frac{p(\mu_{v}^{M}(\Theta,K))^2}{2(1-p)\sigma^{2}}dv \right\} | \mathcal{G}^{(2)}_{t} ]
\end{equation}
the optimal wealth process writes as
\begin{equation}
\hat X^{(2)}_{t}= \frac{X_{0}}{H_{0}}(L_{t})^{\frac{1}{p-1}}H_{t}.
\label{eq:wealth_intermsof_H}
\end{equation}
In order to find the dynamics of $(H_{t})_{t \in[0,T]}$, we first remark  that $(M_{t}:=H_{t}D_{t})_{t \in[0,T]}$ is a $(\mathbb{G}^{(2)},\mathbb{P})$-martingale, where
\begin{equation}
D_{t}:=\exp\left\{ \int_{0}^{t} \frac{p\mu_{v}^{M}(\Theta,K)}{(1-p)\sigma}dW_{v} +\int_{0}^{t} \frac{p(\mu_{v}^{M}(\Theta,K))^2}{2(1-p)\sigma^{2}}dv \right\}.
\label{eq:D_process}
\end{equation}
 By the martingale representation theorem there exists a $\mathbb{G}^{(2)}$-adapted  process $Z^{M}$ such that
\begin{equation}
M_{t}=M_{0}+\int_{0}^{t} M_{v}Z^{M}_{v} dW_{v}.
\end{equation}
From equation (\ref{eq:D_process})
\begin{equation}
d(\frac{1}{D_{t}})=(\frac{1}{D_{t}})\left\{ \left(\frac{p^{2}(\mu_{t}^{M}(\Theta,K))^{2}}{2(1-p)^{2}\sigma^{2}}-\frac{p(\mu_{t}^{M}(\Theta,K))^{2}}{2(1-p)\sigma^{2}}\right)dt
-\frac{p\mu_{t}^{M}(\Theta,K)}{(1-p)\sigma}dW_{t} \right\}
\end{equation}
which leads to the following dynamics for the process
$(H_{t})_{t \in[0,T]}$
$$
dH_{t} =H_{t}\left\{ \left( \frac{p^2(\mu_{t}^{M}(\Theta,K))^{2}}{2(1-p)^{2}\sigma^{2}}-\frac{p(\mu_{t}^{M}(\Theta,K))^{2}}{2(1-p)\sigma^{2}}- \frac{p\mu_{t}^{M}(\Theta,K)}{(1-p)\sigma}Z^{M}_{t}\right)dt +\left( Z^{M}_{t}- \frac{p\mu_{t}^{M}(\Theta,K)}{(1-p)\sigma}\right)dW_{t} \right\}.$$
Denoting $Z^{H}_{t}: =H_{t}Z^{M}_{t}- \frac{p\mu_{t}^{M}(\Theta,K)}{(1-p)\sigma}$
and using the terminal condition $H_{T}=1$, $(H_{t})_{t \in[0,T]}$   satisfies the  following BSDE
\begin{equation}
H_{t}= 1+\int_{t}^{T}\left( \frac{p(\mu_{v}^{M}(\Theta,K))^2}{2(1-p)^{2}\sigma^{2}}H_{v}+\frac{p\mu_{v}^{M}(\Theta,K)}{(1-p)\sigma} Z^{H}_{v} \right)dv -\int_{t}^{T}Z^{H}_{v}dW_{v}.
\label{eq:H_BSDE}
\end{equation}
Thus the dynamics of the optimal wealth process, using (\ref{eq:wealth_intermsof_H}), are
\begin{equation}
d\hat X^{(2)}_{t} = \hat X^{(2)}_t\left( \frac{\left(\mu_{t}^{M}(\Theta,K)\right)^{2}}{(1-p)\sigma^{2}}+\frac{ \mu_{t}^{M}(\Theta,K) Z^{H}_{t}}{ \sigma H_{t}} \right) dt + \hat X^{(2)}_{t}\left( \frac{\mu_{t}^{M}(\Theta,K)}{(1-p)\sigma}+\frac{ Z^{H}_{t}}{ H_{t}} \right)dW_{t}.
\label{eq:optimal_wealth_power_SDE}
\end{equation}
that leads to the  optimal strategy (by comparing (\ref{eq:optimal_wealth_power_SDE}) with (\ref{eq:wealth_fully_informed}))
\begin{equation}
\hat \pi^{(2)}_{t}=\frac{\mu_{t}^{M}(\Theta,K)}{(1-p)\sigma^2}+\frac{ Z^{H}_{t}}{\sigma H_{t}}.
\label{eq:optimal_strategy_power_fully_informed}
\end{equation}
We decompose the time horizon $[0,T]$ into two random time intervals $\lbr 0,\tau\wedge T \lbr $ and $\lbr \tau\wedge T, T \rbr$. On the random interval  $\lbr \tau\wedge T, T \rbr$, the fully informed investor observes the drift term  $\mu^{M}$  thus the  BSDE  (\ref{eq:H_BSDE})
 can be solved explicitly on $\lbr \tau\wedge T, T \rbr$:
\begin{align}
H_{t} &= \exp \left\{\int_{t}^{T}\frac{p(\mu_{v}^{I}(\tau,\Theta,K))^2}{2(1-p)^{2}\sigma^{2}} dv \right\}, \label{eq:H_after_liquidation_solution}\\
Z^{H}_{t} &= 0. \label{eq:Z_H_after_liquidation_solution}
\end{align}
Recalling  (\ref{eq:mu_global_expression}) and using (\ref{eq:H_after_liquidation_solution})-(\ref{eq:Z_H_after_liquidation_solution})  we may decompose the optimal strategy in (\ref{eq:optimal_strategy_power_fully_informed}) into two parts:
\begin{equation*}
\hat\pi_t^{(2)}=1_{\{ 0 \leq t < \tau \wedge T  \}}\hat \pi^{(2,b)}_{t}+1_{\{  \tau \wedge T \leq t \leq  T  \}}\hat \pi_{t}^{(2,a)} \quad
(\ref{eq:optimal_strategy_power_fully_informed_decomposed})
\end{equation*}
where
\begin{align}
\hat \pi^{(2,b)}_{t} & = \frac{\mu}{(1-p)\sigma^2}+\frac{ Z^{H}_{t}}{\sigma H_{t}}, \quad   t  \in \lbr 0, \tau\wedge T \lbr ,\quad (\ref{eq:optimal_strategy_power_fully_informed_before_liquidation})\\
\hat \pi^{(2,a)}_{t} & = \frac{\mu_{t}^{I}(\tau,  \Theta,K)}{(1-p)\sigma^2}, \quad t \in \lbr  \tau \wedge T , T \rbr .
\quad (\ref{eq:optimal_strategy_power_fully_informed_after_liquidation})
\end{align}
\end{proof}
The optimal strategy after liquidation in (\ref{eq:optimal_strategy_power_fully_informed_after_liquidation}) is essentially a Merton-type strategy. The part before liquidation in (\ref{eq:optimal_strategy_power_fully_informed_before_liquidation}) is the sum of a Merton strategy and an extra component\footnote{This extra term is called "hedging demand for parameter risk" by Bj\"ork \cite{bjork2010optimal}.} which is
determined by the solution of the BSDE (\ref{eq:H_BSDE_full_before_liquidation}). It is hard to obtain a closed-form solution for the BSDE (\ref{eq:H_BSDE_full_before_liquidation}), however, we may solve the BSDE (\ref{eq:H_BSDE_full_before_liquidation}) numerically which will be discussed in Section \ref{sect:numerical_result}.

We next consider the case of logarithmic utility for the fully informed investor.

\subsection{Logarithmic utility}
In this section we consider the logarithmic utility $U(x)=\ln(x)$. Using the fact that $I(x) = \frac{1}{x}$ and by (\ref{eq:optimal_terminal_wealth_fully_informed})-(\ref{eq:lambda_solution_fully_informed}) we obtain the optimal terminal wealth
\begin{equation}
\hat X^{(2)}_{T}=\frac{X_{0}}{L_{T}}
\label{eq:optimal_terminal_wealth_log_fully_informed}
\end{equation}
where $L_{T}$ is given by (\ref{eq:likelihood_process_fully_informed}). The optimal expected utility is
\begin{equation*}
V^{(2)}_{0}(\Theta,K)=\ln(X_{0})-\mathbb{E}\left[ \ln(L_{T}) \right].
\end{equation*}
Applying the abstract Bayes' formula to (\ref{eq:optimal_wealth_process_fully_informed}) and using  (\ref{eq:optimal_terminal_wealth_log_fully_informed}), we obtain
\begin{align}
\hat X^{(2)}_{t} =& \mathbb{E}^{\mathbb{Q}}\left[\hat X^{(2)}_T|\mathcal{G}^{(2)}_{t}\right] \nonumber \\
=& \frac{1}{L_{t}} \mathbb{E}\left[\hat X^{(2)}_{T}L_{T}|\mathcal{G}^{(2)}_{t}\right] \nonumber \\
=&  \frac{X_{0}}{L_{t}} \label{eq:optimal_termimal_wealth_bayes_log_fully_informed}
\end{align}
whose dynamics are given by
\begin{equation}
d\hat X^{(2)}_{t}= \frac{\left(\mu^{M}_{t}(\Theta,K)\right)^{2}}{2\sigma^{2}}dt+\frac{\mu^{M}_{t}(\Theta,K)}{\sigma}dW_{t}.
\label{eq:optimal_wealth_log_SDE}
\end{equation}
Comparing (\ref{eq:optimal_wealth_log_SDE}) with (\ref{eq:wealth_fully_informed}) we obtain the optimal strategy
\begin{equation*}
\hat \pi^{(2)}_{t}=\frac{\mu_t^M(\Theta,K)}{\sigma^{2}}.
\end{equation*}
Recalling  (\ref{eq:mu_global_expression}) we may decompose the optimal strategy into two parts
\begin{equation}
\hat\pi_{t}^{(2)}=1_{\{ 0 \leq t < \tau \wedge T  \}}\hat \pi^{(2,b)}_t+1_{\{ \tau\wedge T \leq t  \leq  T  \}}\hat \pi_{t}^{(2,a)}
\label{eq:optimal_strategy_log_fully_informed}
\end{equation}
where
\begin{align}
\hat \pi^{(2,b)}_t & = \frac{\mu}{\sigma^{2}}, \quad t \in \lbr 0, \tau \wedge T \lbr , \label{eq:optimal_strategy_fully_log_before} \\
\hat \pi^{(2,a)}_{t} & = \frac{\mu^{I}_{t}(\tau,\Theta,K)}{\sigma^{2}}, \quad t \in \lbr \tau\wedge T ,   T \rbr. \label{eq:optimal_strategy_fully_log_after}
\end{align}
The optimal trading strategy for the fully informed investor is composed of two Merton strategies before  and after-liquidation.
Accordingly we decompose the optimal wealth process $\hat{X}^{(2)}_{t}$  as
\begin{equation*}
\hat X^{(2)}_{t}=1_{\{ 0 \leq t < \tau \wedge T  \}}\hat{X}^{(2,b)}_t+1_{\{ \tau\wedge T \leq t  \leq  T  \}}\hat{X}^{(2,a)}_{t}
\end{equation*}
where $\hat{X}^{(2,b)}$ and $\hat{X}^{(2,a)}$ satisfy the following SDEs
\begin{align}
d\hat X_{t}^{(2,b)} & = \hat X_{t}^{(2,b)} \hat\pi^{(2,b)}_{t} (\mu_{t} dt+\sigma dW_{t}), & t \in \lbr 0, \tau \wedge T \lbr, 	\label{eq:wealth_fully_informed_log_before_liquidation_SDE}\\
d\hat X_{t}^{(2,a)} & = \hat X_{t}^{(2,a)}\hat{\pi}^{(2,a)}_{t} \left\{ \mu^{I}_{t}(\tau,\Theta,K)dt+ \sigma d W_{t}\right\}, & t \in \lbr \tau\wedge T ,  T \rbr.
\label{eq:wealth_fully_informed_log_after_liquidation_SDE}
\end{align}
Then we decompose the expected utility of terminal wealth into two parts depending on if liquidation occurs before or after time $T$:
\begin{equation}
V^{(2)}_{0}(\Theta,K)=\mathbb{E}[1_{\{\tau>T\}}\ln (\hat X^{(2,b)}_{T})|\mathcal G_0^{(2)} ]+\mathbb{E}[1_{\{\tau\leq T\}}\ln (\hat X^{(2,a)}_{T})|\mathcal G_0^{(2)} ].
\label{eq:utility_decomposition_log_fully_informed}
\end{equation}
The two conditional expectations in (\ref{eq:utility_decomposition_log_fully_informed}) are calculated in Lemma \ref{lemma:optimal_utility_log_term_one} and \ref{lemma:optimal_utility_log_term_two} respectively. Combining those lemmas we obtain the following result.

\begin{proposition}\label{theorem:optimal_utility_expression_log}
The optimal log expected utility for fully informed investors is
\begin{align*}
& V^{(2)}_{0}(\Theta,K)= \\
&\left\{ \mathcal{N}\left(\frac{-\frac{\ln \alpha}{\sigma}+(\frac{\mu}{\sigma}-\frac{1}{2}\sigma) T}{\sqrt{T}}\right)-\exp\left(\frac{2\mu}{\sigma^{2}}-\ln\alpha\right)\mathcal{N}\left(\frac{\frac{\ln \alpha}{\sigma}+(\frac{\mu}{\sigma}-\frac{1}{2}\sigma) T}{\sqrt{T}}\right) \right\} \times \left( \ln(X_{0})+\frac{1}{2}(\mu-\frac{\mu^{2}}{\sigma^{2}})T  \right) \\
&+ \int_{\frac{\ln\alpha}{\sigma}}^{0}\int_{y}^{\infty} \frac{2\mu x(x-2y)}{\sqrt{2\pi T^{3}}}\exp\left\{(\frac{\mu}{\sigma}-\frac{1}{2}\sigma) x-\frac{1}{2}(\frac{\mu}{\sigma}-\frac{1}{2}\sigma)^{2}T-\frac{1}{2T}(2y-x)^{2}\right \}dxdy \\
&-\frac{\ln\alpha}{\sigma}  \int_{0}^{T}\frac{1}{\sqrt{2\pi t^{3}}}\exp\left\{-\frac{1}{2t}\left(\frac{\ln\alpha}{\sigma}-(\frac{\mu}{\sigma}-\frac{1}{2}\sigma) t\right)^{2}\right\}h^{(2)}(t,\Theta,K)dt
\end{align*}	
where
\begin{equation*}
h^{(2)}(t;\theta,k):=\ln X_{0}+\frac{\mu\ln\alpha}{\sigma^{2}} +\frac{\mu}{2}t-\frac{\mu^{2}}{2\sigma^{2}}t +\int_{t}^{T}\frac{\left(\mu^{I}_{v}(t,\theta,k)\right)^{2}}{2\sigma^2}dv.
\end{equation*}
\end{proposition}

In the next section we consider the optimization problem for the partially informed investors.

\section{Partially informed investors}\label{sect:partially_informed}

The portfolio strategy for partially informed investors is supposed to be $\mathbb{G}^{(1)}$-adapted and denoted by $\pi^{(1)} = (\pi^{(1)}_t, 0\leq t\leq T)$.  The wealth process evolves as
\begin{equation}
dX_t^{(1)}=X_t^{(1)}\pi_t^{(1)}(\mu^M_t(\Theta,K) dt+\sigma dW_t), \quad \quad 0 \leq t \leq T.
\label{eq:wealth_partially_informed_partial_observation}
\end{equation}
Similar to (\ref{eq:admissible_strategy_integration_condition}), the admissible strategy set $\mathcal{A}^{(1)}$ is a collection of $\pi^{(1)}$
such that, for any $(\theta,k)\in (0,+\infty)\times (0,1)$,
\begin{equation}
\int_{0}^{T} | \pi^{(1)}_{t}\mu^{M}_{t}(\theta,k) | dt +\int_{0}^{T}|\pi^{(1)}_{t}\sigma|^{2}dt< \infty.
\label{eq:admissible_strategy_integration_condition_partial}
\end{equation}
The portfolio optimization problem for partially informed investors is
\begin{equation}
V^{(1)}_0 = \sup_{\pi^{(1)} \in \mathcal{A}^{(1)}} \mathbb{E} \left[U\left(X^{(1)}_{T}\right)\right].
\label{eq:optimization_objective_partially_informed}
\end{equation}
Note that the optimization problem (\ref{eq:wealth_partially_informed_partial_observation})-(\ref{eq:optimization_objective_partially_informed}) is the case of partial observations since the drift term in (\ref{eq:wealth_partially_informed_partial_observation}) is not $\mathbb{G}^{(1)}$-adapted.

Following \citet{karatzas1998bayesian} we first reduce the optimization problem of partial observation to the case of complete observation.
Recall that the probability measure $\mathbb{Q}$ is defined as
\begin{equation}
\left.\frac{d \mathbb{Q}}{d \mathbb{P}}\right|_{\mathcal{G}^{(2)}_{T}} = L_{T}.
\label{eq:measure_Q_partial}
\end{equation}
with the density process
\begin{equation}
L_{t} = \exp\left\{  -\int_0^t \frac{\mu^{M}_v(\Theta,K)}{\sigma}dW_{v} -\int_{0}^{t} \frac{\left(\mu^{M}_{v}(\Theta,K)\right)^2}{2\sigma^2} dv  \right\}
\label{eq:likelihood_process_partial_a}
\end{equation}
 being  a $(\mathbb{G}^{(2)},\mathbb{P})$-martingale.

We next define the filtered estimate of the drift $\mu^{M}_{t}(\Theta,K)$, based on the observation of the market price, by
\begin{equation}
\bar{\mu}^{M}_{t} = \mathbb{E}\left[ \mu^{M}_{t}(\Theta,K) |\mathcal{G}^{(1)}_{t}\right].
\end{equation}
We define the innovations process $\tilde W$ by
\begin{equation}
d\tilde W_{t}=dW_{t}+\frac{\mu^{M}_t(\Theta,K)-\bar \mu^{M}_{t}}{\sigma}dt,\quad  0\leq t\leq T.
\label{eq:innovation_process_definition}
\end{equation}
By \citet[Lemma 4.1]{bjork2010optimal} we know $\tilde W$ is a standard $(\mathbb{G}^{(1)},\mathbb{P})$-Brownian motion.
Then we may rewrite (\ref{eq:global_risky_asset_dynamics}) as
\begin{equation}
dS_{t}^{M}=S^{M}_{t}\left( \bar{\mu}^{M}_{t}dt+\sigma d\tilde{W}_{t} \right).
\label{eq:global_risky_asset_dynamics_transformed}
\end{equation}
and the wealth process $X^{(1)}$ as
\begin{equation}
dX_t^{(1)}=X_t^{(1)}\pi_t^{(1)}(\bar\mu^M_t dt+\sigma d\tilde W_t), \quad \quad 0 \leq t \leq T
\label{eq:wealth_partially_informed}
\end{equation}
with initial wealth $x_0 \in ]0, + \infty[$. Now the dynamics of the wealth process $X^{(1)}$ is within the framework of a full observation model since $\bar\mu^{M}$ is $\mathbb{G}^{(1)}$-adapted.

Similar to the case of fully informed investors, the optimization problem (\ref{eq:optimization_objective_partially_informed}) can be solved by the martingale approach.
\begin{proposition}
{\rm(i)}The optimal terminal wealth of a partially informed investors, with utility function $U$ and $I=(U^\prime)^{-1}$   is given by
$$\hat X^{(1)}_{T}=I(\lambda \bar L_{T}).$$
The Lagrange multiplier $\lambda$  is determined by the budget constraint
\begin{equation}
\mathbb{E}\left[ I(\lambda \bar L_{T}) \bar L_{T}\right]=x_{0}
\end{equation}
and $\bar L_{T}$ is the density of the risk neutral probability measure for the filtration $\mathbb G^{(1)}$
 \begin{equation}
 \bar L_t=
\exp\left\{ -\int_0^t\frac{\bar\mu^M_{v}}{\sigma}d\tilde W_{v}-\int_{0}^{t}\frac{\left(\bar\mu^{M}_{v}\right)^{2}}{2\sigma^{2}}dv \right\}
\label{eq:likelihood_process_partially_informed}
\end{equation}
 with $\bar\mu^{M}_t = \mathbb{E}\left[\mu^M_{t}(\Theta,K) | \mathcal{G}^{(1)}_{t}\right]$

\noindent{\rm(ii)} The filtered drift estimate $\bar\mu^M$ can be computed as
\begin{equation}
\bar{\mu}^{M}_{t} =
\begin{cases}
\mu, & t \in \lbr 0, \tau \wedge T \lbr \\
\frac{\int_{0}^{\infty}\int_{0}^{1}\left\{ \mu^{M}_{t}(\theta,k)\exp\left\{  \int_0^t \frac{\mu^{M}_v(\theta,k)}{\sigma}dW^{\mathbb{Q}}_{v} -\int_{0}^{t} \frac{\left(\mu^{M}_{v}(\theta,k)\right)^2}{2\sigma^2} dv  \right\} \right\} \varphi(\theta,k)  d\theta dk}{\int_{0}^{\infty}\int_{0}^{1}\left\{\exp\left\{  \int_0^t \frac{\mu^{M}_v(\theta,k)}{\sigma}dW^{\mathbb{Q}}_{v} -\int_{0}^{t} \frac{\left(\mu^{M}_{v}(\theta,k)\right)^2}{2\sigma^2} dv  \right\}\right\} \varphi(\theta,k) d\theta dk}, & t \in \lbr \tau \wedge T, T \rbr .
\end{cases}
\end{equation}
\end{proposition}

\begin{proof}
{\rm(i)} We define the process $\bar{L}_{t} = \mathbb{E}^{\mathbb{Q}}[L_{t}|\mathcal{G}^{(1)}_{t}]$. By rewriting $L_{t}$ in (\ref{eq:likelihood_process_partial_a})
as
\begin{equation}
L_{t} = \exp\left\{  -\int_0^t \frac{\mu^{M}_v(\Theta,K)}{\sigma}d{W}^{\mathbb{Q}}_{v} +\int_{0}^{t} \frac{\left(\mu^{M}_{v}(\Theta,K)\right)^2}{2\sigma^2} dv  \right\},
\end{equation}
we compute
\begin{equation}
\bar{L}_{t} = \exp\left\{ -\int_0^t\frac{\bar\mu^M_{v}}{\sigma}d W^{\mathbb{Q}}_{v}+\int_{0}^{t}\frac{\left(\bar\mu^{M}_{v}\right)^{2}}{2\sigma^{2}}dv \right\} \label{eq:likelihood_process_partial_b}
\end{equation}
Combining (\ref{eq:girsanov_brownian_change_partially_informed}) and (\ref{eq:innovation_process_definition}) we have
\begin{equation}
dW^{\mathbb{Q}}_{t}  = d\tilde{W}_{t} + \frac{\bar{\mu}^{M}_{t}}{\sigma}dt.
\label{eq:brownian_Q_to_innovation}
\end{equation}
Substituting (\ref{eq:brownian_Q_to_innovation}) into (\ref{eq:likelihood_process_partial_b}) we find
\begin{equation}
\bar{L}_{t} = \exp\left\{ -\int_0^t\frac{\bar\mu^M_{v}}{\sigma}d \tilde{W}_{v}-\int_{0}^{t}\frac{\left(\bar\mu^{M}_{v}\right)^{2}}{2\sigma^{2}}dv \right\}.\label{eq:likelihood_process_partial_c}
\end{equation}
Since $\bar{L}_{t}$ is a $(\mathbb{G}^{(1)},\mathbb{P})$-martingale, we define the the risk neutral probability measure $\tilde{\mathbb{Q}}$ by
\begin{equation}
\left.\frac{d \tilde{\mathbb{Q}}}{d \mathbb{P}}\right|_{\mathcal{G}^{(1)}_{t}} = \bar{L}_{t}.
\label{eq:risk_neutral_measure_partial}
\end{equation}
By the fact that $\tilde{W}$ is a $(\mathbb{G}^{(1)},\mathbb{P})$-Brownian motion and the Girsanov's theorem, the process $W^{\tilde{\mathbb{Q}}}$ defined as
\begin{equation}
W^{\tilde{\mathbb{Q}}}_{t}= \tilde{W}_{t} + \int_{0}^{t}\frac{\bar\mu^{M}_{v}}{\sigma}dv, \quad 0\leq t \leq T
\label{eq:girsanov_brownian_change_partially_informed}
\end{equation}
is a $(\mathbb{G}^{(1)}, \tilde{\mathbb{Q}})$-Brownian motion.

Following the same procedure as in Section \ref{sect:fully_informed} we find the optimal terminal wealth $\hat X^{(1)}_{T}$  given by
\begin{equation}
\hat X^{(1)}_{T}=I(\lambda \bar L_{T}),
\end{equation}
where $I=(U^\prime)^{-1}$ and the Lagrange multiplier $\lambda$ is determined by
\begin{equation}
\mathbb{E}\left[ I(\lambda \bar L_{T}) \bar L_{T}\right]=x_{0}.
\end{equation}

{\rm(ii)} Recall that
\begin{equation}
\frac{1}{L_{t}}=\left.\frac{d \mathbb{P}}{d \mathbb{Q}}\right|_{\mathcal{G}^{(2)}_{t}} =   \exp\left\{  \int_0^t \frac{\mu^{M}_v(\Theta,K)}{\sigma}dW^{\mathbb{Q}}_{v} -\int_{0}^{t} \frac{\left(\mu^{M}_{v}(\Theta,K)\right)^2}{2\sigma^2} dv  \right\}
\end{equation}
is $(\mathbb{G}^{(2)},\mathbb{Q})$-martingale. By the Kallianpur-Striebel formula, which is related to Bayes formula, we have
\begin{align}
\bar{\mu}^{M}_{t} = & \mathbb{E}\left[ \mu^{M}_{t}(\Theta,K) |\mathcal{G}^{(1)}_{t}\right] \nonumber\\
=&  \frac{\mathbb{E}^{\mathbb{Q}}\left[ \mu^{M}_{t}(\Theta,K)L_{T}  |\mathcal{G}^{(1)}_{t}\right]}{\mathbb{E}^{\mathbb{Q}}[L_{t}|\mathcal{G}^{(1)}_{t}]}  \nonumber \\
=& \frac{\mathbb{E}^{\mathbb{Q}}\left[ \mathbb{E}^{\mathbb{Q}}\left[\mu^{M}_{t}(\Theta,K)L_{T} |\mathcal{G}^{(2)}_{t}\right]\Big|\mathcal{G}^{(1)}_{t}\right]}{\mathbb{E}^{\mathbb{Q}}[L_{t}|\mathcal{G}^{(1)}_{t}]}\nonumber \\
=& \frac{\mathbb{E}^{\mathbb{Q}}\left[ \mu^{M}_{t}(\Theta,K)L_{t}  |\mathcal{G}^{(1)}_{t}\right]}{\mathbb{E}^{\mathbb{Q}}[L_{t}|\mathcal{G}^{(1)}_{t}]}\nonumber \\
=& \frac{\mathbb{E}^{\mathbb{Q}}\left[ \mu^{M}_{t}(\Theta,K)\exp\left\{  \int_0^t \frac{\mu^{M}_v(\Theta,K)}{\sigma}dW^{\mathbb{Q}}_{v} -\int_{0}^{t} \frac{\left(\mu^{M}_{v}(\Theta,K)\right)^2}{2\sigma^2} dv  \right\}  |\mathcal{G}^{(1)}_{t}\right]}{\mathbb{E}^{\mathbb{Q}}[\exp\left\{  \int_0^t \frac{\mu^{M}_v(\Theta,K)}{\sigma}dW^{\mathbb{Q}}_{v} -\int_{0}^{t} \frac{\left(\mu^{M}_{v}(\Theta,K)\right)^2}{2\sigma^2} dv  \right\}|\mathcal{G}^{(1)}_{t}]}.
\end{align}
Since the measure $\mathbb{Q}$ coincides with $\mathbb{P}$ on $\mathcal{G}^{(2)}_{0} = \sigma(\Theta, K)$, the distribution of $(\Theta,K)$ under $\mathbb{Q}$ is identical to the one under $\mathbb{P}$.
Recall that the Brownian motion $W^{\mathbb{Q}}$ is  independent of $\sigma(\Theta,K)$ we have
\begin{equation}
\bar{\mu}^{M}_{t} =
\frac{\int_{0}^{\infty}\int_{0}^{1}\left\{ \mu^{M}_{t}(\theta,k)\exp\left\{  \int_0^t \frac{\mu^{M}_v(\theta,k)}{\sigma}dW^{\mathbb{Q}}_{v} -\int_{0}^{t} \frac{\left(\mu^{M}_{v}(\theta,k)\right)^2}{2\sigma^2} dv  \right\} \right\} \varphi(\theta,k)  d\theta dk}{\int_{0}^{\infty}\int_{0}^{1}\left\{\exp\left\{  \int_0^t \frac{\mu^{M}_v(\theta,k)}{\sigma}dW^{\mathbb{Q}}_{v} -\int_{0}^{t} \frac{\left(\mu^{M}_{v}(\theta,k)\right)^2}{2\sigma^2} dv  \right\}\right\} \varphi(\theta,k) d\theta dk}.
\label{eq:filter_estimate}
\end{equation}
For $t<\tau\wedge T$ we have $\bar{\mu}^{M}_{t} = \mu$ due to the fact that $\mu^{M}_{t} = \mu$.
\end{proof}

Following \citet{Fuji} there exits a martingale representation theorem with respect to the $(\mathbb{G}^{(1)},\tilde{\mathbb{Q}})$-Brownian motion $W^{\tilde{\mathbb{Q}}}$. Similar to the case of fully informed investors, the optimal strategy $\pi^{(1)}$ relies on the martingale representation theorem. For a general utility function, the optimal strategy $\pi^{(1)}$ does not have explicit expression.
In the next subsections,  we will consider power  and logarithmic utilities.

\subsection{Power utility}
We first consider the power utility  $U(x)=\frac{x^p}{p}$, $0<p<1$.  The optimal terminal wealth at $T$ is given by
\begin{equation}
\hat X^{(1)}_{T}=\frac{x_{0}}{\mathbb{E}\left[\left(\bar{L}_{T}\right)^{\frac{p}{p-1}}\right]}\left(\bar{L}_{T}\right)^{\frac{1}{p-1}}
\end{equation}
where $\bar{L}_{T}$ is given by equation (\ref{eq:likelihood_process_partially_informed}).
The optimal expected utility is
\begin{equation}
V^{(1)}_{0}=\frac{(x_{0})^{p}}{p}\left(\mathbb{E}\left[\left(\bar{L}_{T}\right)^{\frac{p}{p-1}}\right]\right)^{1-p}.
\label{eq:optimal_utility_power_partially_informed}
\end{equation}

Similar to the case of fully informed investors, we may decompose the optimal strategy $\hat{\pi}^{(1)}$ into two parts:
\begin{equation}
\hat\pi_t^{(1)}=1_{\{ 0 \leq t < \tau \wedge T  \}}\hat \pi^{(1,b)}_{t}+1_{\{  \tau \wedge T \leq t \leq  T  \}}\hat \pi_{t}^{(1,a)}.
\label{eq:optimal_strategy_power_partially_informed_decomposed}
\end{equation}
Following a similar procedure as in Section \ref{sect:fully_power_utility} we obtain
\begin{align}
\hat \pi^{(1,b)}_{t} & = \frac{\mu}{(1-p)\sigma^2}+\frac{ Z^{\bar{H}}_{t}}{\sigma \bar{H}_{t}}, \quad  t \in \lbr 0,  \tau\wedge T \lbr, \label{eq:optimal_strategy_power_partially_informed_before_liquidation}\\
\hat \pi^{(1,a)}_{t} & = \frac{\bar{\mu}_{t}^{I}}{(1-p)\sigma^2}, \quad t \in \lbr  \tau \wedge T , T \rbr .
\label{eq:optimal_strategy_power_partially_informed_after_liquidation}
\end{align}
where $(\bar{H}, Z^{\bar{H}})$ satisfies the linear BSDE
\begin{equation}
\bar H_{t} = 1+\int_{t}^{T}\left( \frac{p\left(\bar\mu^{M}_v\right)^{2}}{2(1-p)^{2}\sigma^{2}}\bar{H}_{v}+\frac{p\bar\mu^{M}_v}{(1-p)\sigma} Z^{\bar{H}}_{v} \right)dv-\int_{t}^{T}Z^{\bar{H}}_{v}d\tilde{W}_{v}.
\label{eq:H_BSDE_partial_before_liquidation}
\end{equation}
We will discuss the numerical solution of BSDE (\ref{eq:H_BSDE_partial_before_liquidation}) in  Section \ref{sect:numerical_result}.

\subsection{Log utility}
In this section we consider the logarithmic utility $U(x)=\ln(x)$.  The optimal terminal wealth at $T$ is given by
\begin{equation}
\hat X^{(1)}_{T}=\frac{x_{0}}{\bar{L}_{T}}.
\end{equation}
The optimal expected utility is
\begin{equation}
V^{(1)}_{0}=\ln(x_{0})-\mathbb{E}\left[ \ln(\bar{L}_{T}) \right].
\end{equation}
The optimal investment process $\hat{\pi}^{(1)}$  is given by
\begin{equation}
\hat{\pi}_{t}^{(1)}=1_{\{0 \leq t < \tau \wedge T \}}\hat \pi^{(1,b)}_{t}+1_{\{\tau \wedge T \leq t \leq  T  \}}\hat \pi_t^{(1,a)}
\end{equation}
where
\begin{align}
\hat \pi^{(1,b)}_{t} & = \frac{\mu}{\sigma^{2}}, \quad t \in \lbr 0,  \tau\wedge T \lbr, \label{eq:optimal_strategy_partially_log_before}\\
\hat \pi^{(1,a)}_t & = \frac{\bar\mu^{I}_{t}}{\sigma^{2}}, \quad t \in \lbr  \tau \wedge T , T \rbr .\label{eq:optimal_strategy_partially_log_after}
\end{align}
 We decompose the optimal wealth process into before and after liquidation parts as
\begin{equation}
\hat{X}^{(1)}_t=1_{\{0 \leq t < \tau \wedge T \}}\hat X^{(1,b)}_{t}+1_{\{\tau \wedge T \leq t \leq  T \}}\hat X^{(1,a)}_{t}
\end{equation}
where $\hat X^{(1,b)}_{t}$ and $\hat X^{(1,a)}_{t}$ satisfy the following SDEs
\begin{align}
d\hat X_{t}^{(1,b)}&=\hat X_{t}^{(1,b)} \hat\pi^{(1,b)}(\mu dt+ \sigma dW_{t}), & t \in \lbr 0,  \tau\wedge T \lbr, \\
d\hat X_{t}^{(1,a)}&=\hat X_{t}^{(1,a)}\hat{\pi}_{t}^{(1,a)}\left\{ \bar\mu^{a}_{t}dt+\sigma d\tilde{W}_{t}\right\}, & t \in \lbr  \tau \wedge T , T \rbr .
\end{align}
Then we decompose the expected utility of terminal wealth $V^{(1)}_{0}$ into two parts depending on if liquidation occurs before or after time $T$:
\begin{equation}
V^{(1)}_{0}=\mathbb{E}\left[1_{\{T<\tau\}}\ln \left(\hat X^{(1,b)}_T\right)\right]+\mathbb{E}\left[1_{\{T\geq \tau\}}\ln \left(\hat X^{(1,a)}_T\right)\right].
\label{eq:utility_decomposition_log_partially_informed}
\end{equation}
Comparing (\ref{eq:optimal_strategy_fully_log_before}) and (\ref{eq:optimal_strategy_partially_log_before}), we know partially informed investors holds the same optimal strategy as the fully informed investor before liquidation. The optimal terminal wealth for partially and fully informed investors  are identical if no liquidation occurs before $T$, that is
$$\mathbb{E}\left[1_{\{T<\tau\}}\ln \left(\hat X^{(1,b)}_T\right)\right] = \mathbb{E}\left[1_{\{T<\tau\}}\ln \left(\hat X^{(2,b)}_T\right)\right].$$
Thus the first expectation in (\ref{eq:utility_decomposition_log_partially_informed}) has been calculated in Lemma \ref{lemma:optimal_utility_log_term_one} and the other expectation is calculated in Lemma \ref{lemma:optimal_utility_partially_log_term_two}.
Combining those lemmas we obtain the following result.
\begin{proposition}\label{th:opt_V_log_partial} The  optimal log expected utility for the fully informed investor is
\begin{align}
\notag & V^{(1)}_0=\\
\notag &\left\{ \mathcal{N}\left(\frac{-\frac{\ln \alpha}{\sigma}+(\frac{\mu}{\sigma}-\frac{1}{2}\sigma) T}{\sqrt{T}}\right)-\exp\left(\frac{2\mu}{\sigma^2}-\ln\alpha\right)\mathcal{N}\left(\frac{\frac{\ln \alpha}{\sigma}+(\frac{\mu}{\sigma}-\frac{1}{2}\sigma) T}{\sqrt{T}}\right) \right\} \times \left( \ln(x_0)+\frac{1}{2}(\mu-\frac{\mu^2}{\sigma^2})T  \right)\nonumber \\
\notag &+ \int_{\frac{\ln\alpha}{\sigma}}^{0}\int_y^\infty \frac{2\mu x(x-2y)}{\sqrt{2\pi T^3}}\exp\left\{(\frac{\mu}{\sigma}-\frac{1}{2}\sigma) x-\frac{1}{2}(\frac{\mu}{\sigma}-\frac{1}{2}\sigma)^2T-\frac{1}{2T}(2y-x)^2\right \}dxdy\\
\notag &-\frac{\ln\alpha}{\sigma}\int_0^T\frac{1}{\sqrt{2\pi t^3}}\exp\left\{-\frac{1}{2t}\left(\frac{\ln\alpha}{\sigma}-(\frac{\mu}{\sigma}-\frac{1}{2}\sigma) t\right)^2\right\}h^{(1)}(t)dt
\end{align}	
where
\begin{equation*}
h^{(1)}(t):=\ln x_0+\frac{\mu\ln\alpha}{\sigma^2} +\frac{\mu}{2}t-\frac{\mu^2}{2\sigma^2}t +\int_{t}^{T}\frac{\left(\bar\mu^M_v\right)^2}{2\sigma^{2}}dv.
\end{equation*}
\end{proposition}

We will consider the optimization problem for the  uninformed investors.

\section{Uninformed investors}\label{sect:uninformed}
The uninformed investors erroneously believe the market  price of the asset  follows a Black Scholes dynamics with constant $\mu$. That is, uninformed investors act as Merton investors.
To compare with the fully informed and partially informed investors, we shall consider  both the power utility and logarithmic utility in the following sections.

\subsection{Power Utility}
We first consider the power utility, i.e. $U(x)=\frac{x^p}{p}$. The uninformed investors adopt   the Merton strategy
\begin{equation}
\hat \pi^{(0)}=\frac{\mu}{(1-p)\sigma^2}.
\label{eq:optimal_strategy_power_uninformed}
\end{equation}
However, the market price process of the asset is given by (\ref{eq:global_risky_asset_dynamics}).  Therefore, corresponding to the sub-optimal strategy given by (\ref{eq:optimal_strategy_power_uninformed}), the  wealth process  $\hat X^{(0)}_t$ is written as
$$\hat X^{(0)}_t=1_{\{0\leq t < \tau \wedge T  \}}\hat X^{(0,b)}_t+1_{\{ \tau \wedge T \leq t \leq  T  \}} \hat X^{(0,a)}_t$$
where $ \hat X^{b}_t$ and $\hat X^{(0,a)}_t$ are given by
\begin{align}
d \hat X_t^{(0,b)} & =\hat X_{t}^{(0,b)}\hat\pi^{(0)} \left(\mu dt+\sigma dW_{t}\right), &   t \in\lbr 0,\tau\wedge T \lbr, \\
d \hat X^{(0,a)}_{t}&=\hat X^{(0,a)}_{t}\hat{\pi}^{(0)}\left(\mu^{I}_{t}(\tau,\Theta,K)dt+\sigma dW_{t}\right), & t \in \lbr  \tau \wedge T , T \rbr .
\end{align}

We next compute the expected utility of final wealth $\mathbb{E}[U(\hat X^{0}_T)]$  using the investment strategy given by (\ref{eq:optimal_strategy_power_uninformed}).  We decompose $\mathbb{E}[U(\hat X^{0}_T)]$ into two parts depending on whether or not liquidation occurs before time $T$
\begin{equation}
\mathbb{E}\left[U\left(\hat X^{0}_{T}\right)\right]
=  \mathbb{E}\left[1_{\{\tau>T\}}U\left(\hat X^{(0,b)}_{T}\right)\right] +\mathbb{E}\left[1_{\{\tau\leq T\}}U\left(\hat X_{T}^{(0,a)}\right) \right].
\label{utility_power_decomposition_uninformed}
\end{equation}
The two expectations in (\ref{utility_power_decomposition_uninformed}) are computed in Lemma \ref{lemma:optimal_utility_power_term_one_uninformed}  and \ref{lemma:optimal_utility_power_term_two_uninformed} respectively.
\begin{proposition}
The expected power utility of an uninformed investor who follows the suboptimal strategy (\ref{eq:optimal_strategy_power_uninformed}) is
\begin{align*}
& \mathbb{E}\left[U\left(\hat X^{0}_{T}\right)\right] =\frac{x_{0}^{p}}{p}\exp\left(\frac{p\mu^{2}T}{2(1-p)\sigma^{2}} \right) \\
&\times\left\{\mathcal{N}\left(\frac{-\frac{\ln \alpha}{\sigma}+(\frac{\mu}{(1-p)\sigma}-\frac{\sigma}{2})T}{\sqrt{T}}\right)-\exp\left(\frac{2\mu\ln \alpha}{(1-p)\sigma^{2}}-\ln \alpha\right)\mathcal{N}\left(\frac{\frac{\ln \alpha}{\sigma}+(\frac{\mu}{(1-p)\sigma}-\frac{\sigma}{2})T}{\sqrt{T}}\right)\right\}  \\
& -\frac{\ln\alpha}{\sigma}\int_{0}^{1}\int_{0}^{\infty}\int_{0}^{T}\frac{1}{\sqrt{2\pi t^{3}}}\exp\left\{-\frac{1}{2t}\left(\frac{\ln\alpha}{\sigma}-(\frac{\mu}{\sigma}-\frac{1}{2}\sigma) t\right)^{2}\right\}l^{(0)}(t,\theta,k)\varphi(\theta)dtd\theta dk
\end{align*}
where
\begin{equation*}
l^{(0)}(t,\theta,k)=\frac{x_{0}^{p}}{p}\exp\left\{\frac{\mu\ln\alpha}{(1-p)\sigma^{2}} +\frac{1}{2}\frac{\mu}{(1-p)}t-\frac{1}{2}\frac{\mu^{2}}{(1-p)^{2}\sigma^{2}}t+\int_{t}^{T}\left( \frac{p\mu\mu^{I}_{v}(t,\theta,k)}{(1-p)\sigma^{2}} \right)dv\right\}.
\end{equation*}
\end{proposition}

We next consider the same problem for the uniformed investor under logarithmic utility.

\subsection{Logarithmic Utility}
In case of logarithmic utility, uninformed investors adopt   the Merton strategy
\begin{equation}
\hat \pi^{(0)}=\frac{\mu}{\sigma^2}.
\label{eq:optimal_strategy_log_uninformed}
\end{equation}
We denote by $\hat{X}^{(0)}_{t}$  the wealth process for  uninformed investors as holding the sub-optimal strategy $\hat {\pi}^{(0)}_{t}$ given by (\ref{eq:optimal_strategy_log_uninformed}). Similar to the case of power utility we calculate the expectation $\mathbb{E}[U(\hat X^{(0)}_{T})]$ using the  decomposition
\begin{equation}
E[\ln(\hat X^{(0)}_{T})]= \mathbb{E}[1_{\{\tau>T\}}\ln(\hat X^{(0,b)}_{T})] +\mathbb{E}[1_{\{\tau\leq T\}} \ln(\hat X_{T}^{(0,a)}) ]. \label{eq:utility_log_decomposition_uninformed}
\end{equation}
Comparing (\ref{eq:optimal_strategy_fully_log_before}) and (\ref{eq:optimal_strategy_log_uninformed}), we know  uninformed investors hold the same optimal strategy as the fully informed investors before liquidation. The terminal wealth for uninformed and fully informed investors  are identical if no liquidation occurs before $T$, that is
$$\mathbb{E}\left[1_{\{T<\tau\}}\ln \left(\hat X^{(0,b)}_T\right)\right] = \mathbb{E}\left[1_{\{T<\tau\}}\ln \left(\hat X^{(2,b)}_T\right)\right].$$
Thus the first expectation in (\ref{eq:utility_log_decomposition_uninformed}) has been calculated in Lemma \ref{lemma:optimal_utility_log_term_one} and the other expectation is calculated in Lemma \ref{lemma:optimal_utility_log_term_two_uninformed}.

\begin{proposition}
The expected log utility of an uniformed investor who follows the suboptimal investment strategy (\ref{eq:optimal_strategy_log_uninformed}) is
\begin{align*}
& E[\ln(\hat X^{(0)}_{T})]=  \\
&  \left\{ \mathcal{N}\left(\frac{-\frac{\ln \alpha}{\sigma}+(\frac{\mu}{\sigma}-\frac{1}{2}\sigma) T}{\sqrt{T}}\right)-\exp\left(\frac{2\mu}{\sigma^{2}}-\ln\alpha\right)\mathcal{N}\left(\frac{\frac{\ln \alpha}{\sigma}+(\frac{\mu}{\sigma}-\frac{1}{2}\sigma) T}{\sqrt{T}}\right) \right\} \times \left( \ln(x_{0})+\frac{1}{2}(\mu-\frac{\mu^{2}}{\sigma^{2}})T  \right) \\
&+ \int_{\frac{\ln\alpha}{\sigma}}^{0}\int_y^\infty \frac{2\mu x(x-2y)}{\sqrt{2\pi T^{3}}}\exp\left\{(\frac{\mu}{\sigma}-\frac{1}{2}\sigma) x-\frac{1}{2}(\frac{\mu}{\sigma}-\frac{1}{2}\sigma)^{2}T-\frac{1}{2T}(2y-x)^2\right \}dxdy \\
&-\frac{\ln\alpha}{\sigma}\int_{0}^{1}\int_{0}^{\infty}\int_{0}^{T}\frac{1}{\sqrt{2\pi t^{3}}}\exp\left\{-\frac{1}{2t}\left(\frac{\ln\alpha}{\sigma}-(\frac{\mu}{\sigma}-\frac{1}{2}\sigma) t\right)^{2}\right\}h^{(0)}(t,\theta,k)\varphi(\theta,k)dtd\theta dk
\end{align*}
where
\begin{equation*}
	h^{(0)}(t,\theta,k):=\ln x_{0}+\frac{\mu\ln\alpha}{\sigma^{2}} +\frac{\mu}{2}t-\frac{\mu^{2}}{2\sigma^{2}}t+\int_{t}^{T}\left( \frac{2\mu\mu^{I}_{v}(t,\theta,k)- \mu^2}{2\sigma^{2}} \right)dv
\end{equation*}
\end{proposition}

We next present some numerical results.

\section{Numerical results}\label{sect:numerical_result}

In this section we illustrate numerical results of the optimization problem for the three types of investors.
We set the parameters $ \mu=0.07, \sigma=0.2$ and the initial value $S_{0}=80$. We let the investment horizon $T=1$. The liquidation trigger level is chosen as $\alpha=0.9$.   The stochastic processes are discretized using an Euler scheme with $M=250$ steps and time intervals of length $\Delta t=\frac{1}{250}$.  The number of simulations is $N=10^5$. We suppose the distribution of $(\Theta,K)$ is uniform on $[0.05, 0.15]\times[0.02,0.08]$. The initial wealth is assumed to be $x_{0} = 80$.  The power utility function is specified as $U(x)=2x^{\frac{1}{2}}$.

\subsection{Filtered estimate of the drift}

The time horizon $[0,1]$ is discretized equally as $0=t_{0}<t_{1}<\cdots<t_{M}=1$. For $0\leq m \leq M$ we denote by $\mu^{M}_{t_{m}}(\Theta,K)$ the discretized approximation of $\mu^{M}(\Theta,K)$ at time $t_{m}$.
For $0\leq m \leq M-1$, we denote by $\Delta W_{m}$ the increment of the Brownian motion over the time interval $[t_{m}, t_{m+1}]$. The approximation of the increment of the $(\mathbb{G}^{(2)},\mathbb{Q})$-Brownian motion is $\Delta W^{\mathbb{Q}}_{m} = \Delta W_{m} + \frac{\mu^{M}_{t_{m}}(\Theta,K)}{\sigma}\Delta t$.  We approximate the filtered drift estimate in (\ref{eq:filter_estimate}) at time $t_{m}$ by
\begin{equation}
\hat{\mu}^{M}_{t_{m}} =
\frac{\int_{0}^{\infty}\int_{0}^{1}\left\{ \mu^{M}_{t_{m}}(\theta,k)\exp\left\{  \sum\limits_{0\leq i \leq m-1} \left(\frac{\mu^{M}_{t_{i}}(\theta,k)}{\sigma}( \Delta W_{i} + \frac{\mu^{M}_{t_{i}}(\Theta,K)}{\sigma}\Delta t) -\frac{\left(\mu^{M}_{t_{i}}(\theta,k)\right)^2}{2\sigma^2} \Delta t \right) \right\} \right\} \varphi(\theta,k)  d\theta dk}{\int_{0}^{\infty}\int_{0}^{1}\left\{\exp\left\{  \sum\limits_{0\leq i\leq m-1} \left(\frac{\mu^{M}_{t_{i}}(\theta,k)}{\sigma}( \Delta W_{i} + \frac{\mu^{M}_{t_{i}}(\Theta,K)}{\sigma}\Delta t) -\frac{\left(\mu^{M}_{t_{i}}(\theta,k)\right)^2}{2\sigma^2} \Delta t \right) \right\}\right\} \varphi(\theta,k) d\theta dk}.
\label{eq:filter_estimate_num_a}
\end{equation}
We use Monte-Carlo method to estimate the integral in (\ref{eq:filter_estimate_num_a}). Suppose the number of simulation is $N$. For $1\leq n \leq N$, we denote by $(\theta^{n},k^{n})$ the realized value of the random variable $(\Theta,K)$ in the
$n$th simulation. We estimate $\hat{\mu}^{M}_{t_{m}}$ in (\ref{eq:filter_estimate_num_a}) by the sample mean
\begin{equation}
\tilde{\mu}^{M}_{t_{m}} =
\frac{\sum\limits_{1\leq n\leq N}\left\{ \mu^{M}_{t_{m}}(\theta^{n},k^{n})\exp\left\{  \sum\limits_{0\leq i\leq m-1} \left(\frac{\mu^{M}_{t_{i}}(\theta^{n},k^{n})}{\sigma}( \Delta W_{i} + \frac{\mu^{M}_{t_{i}}(\Theta,K)}{\sigma}\Delta t) -\frac{\left(\mu^{M}_{t_{i}}(\theta^{n},k^{n})\right)^2}{2\sigma^2} \Delta t \right) \right\} \right\}}{\sum\limits_{1\leq n\leq N}\left\{\exp\left\{  \sum\limits_{0\leq m\leq m-1} \left(\frac{\mu^{M}_{t_{i}}(\theta^{n},k^{n})}{\sigma}( \Delta W_{i} + \frac{\mu^{M}_{t_{i}}(\Theta,K)}{\sigma}\Delta t) -\frac{\left(\mu^{M}_{t_{i}}(\theta^{n},k^{n})\right)^2}{2\sigma^2} \Delta t \right) \right\}\right\}}.
\label{eq:filter_estimate_num_b}
\end{equation}

In Figure~\ref{fig:drift:term:compare} we illustrate a sample filter estimate $\tilde{\mu}^{M}$ compared with the drift term $\mu^{M}(\Theta,K)$ in a specific scenario where the realized value of the liquidation random variables are $(\Theta,K)=(0.1,0.05)$. From Figure~\ref{fig:drift:term:compare}  we note that the filtered estimate of the drift is very close to the realized drift. This result suggests that knowing the functional  form of the market impact is more relevant than the actual realization of $(\Theta,K)$.

\begin{figure}[H]
	\centering
	\includegraphics[width=0.5\linewidth]{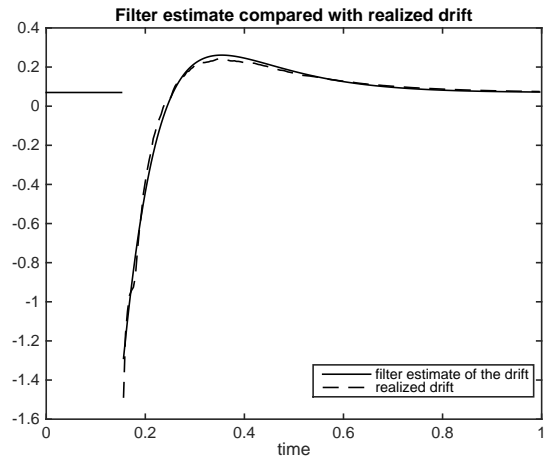}
	\caption{Filter estimate of the drift compared with the realized drift}
	\label{fig:drift:term:compare}
\end{figure}

\subsection{Optimal strategy for power utility}
In this section we illustrate the optimal strategies for fully and partially informed investors in case of power utility by solving the related BSDE numerically. We skip the discussion of log utility since the optimal strategies are simply the "myopic" Merton strategy.
In case of fully informed investors, we approximate the BSDE (\ref{eq:H_BSDE}) by the following discretized BSDE
\begin{align}
\tilde{H}_{t_{m+1}} =& \tilde{H}_{t_{m}}- \left( \frac{p(\mu^{M}_{t_{m}}(\tilde{\theta},\tilde{k}))^{2}}{2(1-p)^{2}\sigma^{2}}\tilde{H}_{t_{m}}+\frac{p\mu^{M}_{t_{m}}(\tilde{\theta},\tilde{k})}{(1-p)\sigma}\tilde Z^{H}_{t_{m}} \right)\Delta t + \tilde Z^{H}_{t_{m}}\Delta W_{t_{m}}, \quad t_{0}\leq t_{m} < t_{M}, \label{eq:H_discretized}\\
\tilde{H}_{t_{M}} =& 1.\label{eq:Z_H_discretized}
\end{align}

The BSDE (\ref{eq:H_discretized})-(\ref{eq:Z_H_discretized}) can be solved using the following recursive scheme (see \citet{gobet2005regression})
\begin{align}
\tilde{Z}^H_{t_{m}}  =& \frac{1}{\Delta t} \mathbb{E}[\tilde{H}_{t_{m+1}}\Delta W_{t_{m}}|\mathcal{G}^{(2)}_{t_{m}}], \label{eq:H_numerical_solution}\\
\tilde{H}_{t_{m}}  =& \frac{\mathbb{E}[\tilde{H}_{t_{m+1}}|\mathcal{G}^{(2)}_{t_{m}}]+ \frac{p\mu^{M}_{t_{m}}(\tilde{\theta},\tilde{k})}{(1-p)\sigma}\tilde{Z}^{H}_{t_{m}}\Delta t}{1-\frac{p(\mu^{M}_{t_{m}}(\tilde{\theta},\tilde{k}))^{2}}{2(1-p)^{2}\sigma^{2}}\Delta t}. \label{eq:Z_numerical_solution}
\end{align}
We estimate the conditional expectation in (\ref{eq:H_numerical_solution}) and (\ref{eq:Z_numerical_solution}) by the Monte-Carlo regression approach proposed by \citet{gobet2005regression}. Note that the market price process $S^{M}_{t}$ is not Markovian with respect to $(\mathbb{G}^{(2)},\mathbb{P})$. We define the running minimum process $\tilde{S}^{M}_{t} = \inf\{  S^{M}_{v}| 0\leq v\leq  t \}$ and note that the pair $(S^{M}_{t}, \tilde{S}^{M}_{t} )$ is Markovian with respect to $(\mathbb{G}^{(2)},\mathbb{P})$. Hence we may choose the regression basis functions: $1,x, x^2, y, y^2$ and $xy$.  By the regression method of \citet{gobet2005regression} the conditional expectations in (\ref{eq:H_numerical_solution}) and (\ref{eq:Z_numerical_solution}) can be estimated by
$$c_{1} + c_{2}(S^{M}_{t}-\alpha S_{0})+c_{3}(S^{M}_{t}-\alpha S_{0})^2+c_{4}(\tilde{S}^{M}_{t}-\alpha S_{0})+c_{5}(\tilde{S}^{M}_{t}-\alpha S_{0})^2+c_{6}(\tilde{S}^{M}_{t}-\alpha S_{0})(\tilde{S}^{M}_{t}-\alpha S_{0})$$ for some coefficients $c_{i}, 1\leq i\leq 6$.

We approximate the optimal strategy for fully informed investor $\hat{\pi}^{(2)}$  by $\tilde{\pi}^{(2,b)}$ as follows
\begin{equation}
\tilde{\pi}^{(2)}_{t_{m}}=\frac{\mu}{(1-p)\sigma^{2}}+\frac{\tilde Z^{H}_{t_{m}}}{\sigma \tilde{H}_{t_{m}}},\quad 0\leq t_{m} \leq  t_{M}.
\end{equation}
Following a similar procedure we may solve the related BSDE for partially informed investors and obtain the approximate optimal strategy.

Figure \ref{fig:optimal_strategy_power_full}  illustrates the approximated optimal strategies for fully and partially investors respectively corresponding to one sample path of the risky asset price where liquidation occurs well before the terminal time $T$.  In particular for the path of the asset price in Figure \ref{fig:optimal_strategy_power_full} liquidation occurs at time $t=0.1540$.  Before liquidation the two strategies are indistinguishable due to the scale.  We plot the optimal strategies before liquidation in  Figure \ref{fig:optimal_strategy_power_Before_Liq} and note that there is some tracking error before liquidation.  This difference may be due to the fact that the before liquidation strategy of both investors contains a component which depends on the solution of a BSDE, which is accomplished backward in time, and in particular depends recursively on the filtered drift estimate for the partially informed investor.  Hence, owing to tracking error typical to filtering problems some errors may be propogated to the before liquidation strategy through the numerical solution procedure for the associated BSDE.   Table \ref{table:optimal_strategies_before_liq} presents the approximated optimal strategies for fully and partially investors at times before liquidation corresponding to Figure \ref{fig:optimal_strategy_power_Before_Liq}.  %

\begin{figure}[H]
\centering
  \includegraphics[width=.75\linewidth]{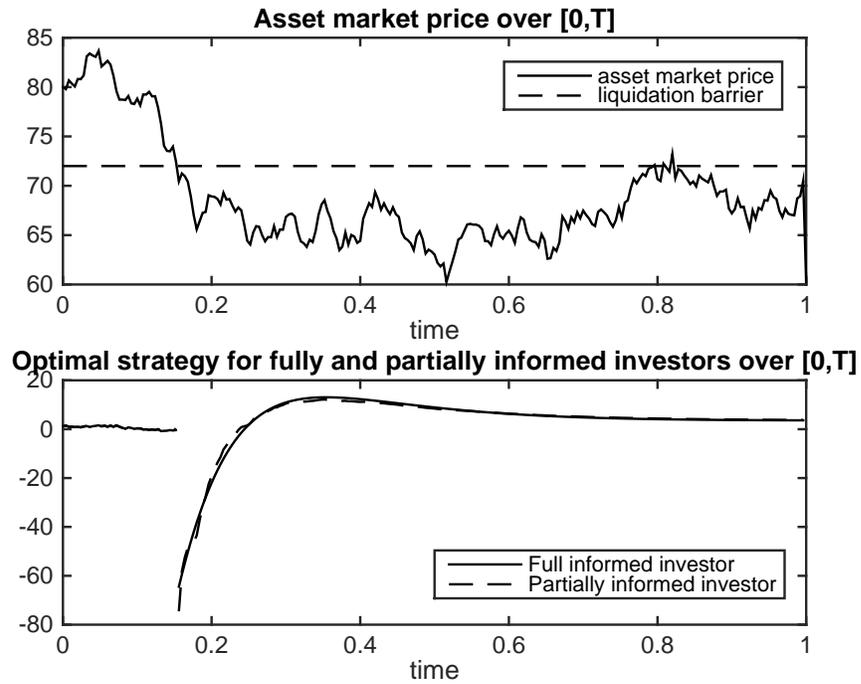}
  \caption{Approximated optimal strategy for fully and partially informed investors over $[0,T]$}
\label{fig:optimal_strategy_power_full}
\end{figure}

\begin{figure}[H]
\centering
  \includegraphics[width=.75\linewidth]{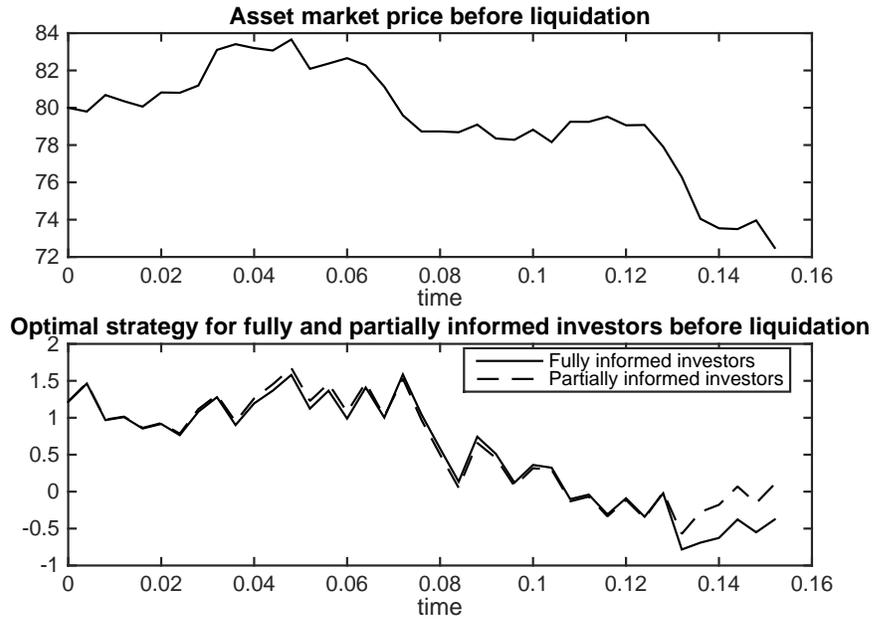}
  \caption{Approximated optimal strategy for fully and partially informed investors before liquidation}
\label{fig:optimal_strategy_power_Before_Liq}
\end{figure}

\begin{table}[H]
	\centering
	\begin{tabular}{|c|c|c|c|c|c|c|c|c|c|}	
	\hline
	  $t_{m}$ &   0.1200 &   0.1240 &   0.1280 & 0.1320 &   0.1360  &  0.1400 &   0.1440  &  0.1480  &  0.1520 \\	
	\hline
	  $S^{M}_{t_{m}}$ &  79.0600  & 79.0766 &  77.9106 & 76.2818 & 74.0479 &  73.5371 &   73.4940  & 73.9593 &   72.4905 \\
	  \hline
	  $\pi^{(1)}_{t_{m}}$ & -0.1127 &  -0.3614  & -0.0063 &  -0.5712  & -0.2756 &  -0.1780 &   0.0699 &   -0.1559  &  0.1043 \\
	  \hline
	   $\pi^{(2)}_{t_{m}}$ & -0.0898  &  -0.3399 &  -0.0224 &  -0.7831  & -0.6907 &  -0.6265 &  -0.3766 &  -0.5502 &  -0.3760 \\
	  \hline
	\end{tabular}
	\caption{Approximated optimal strategies before liquidation}
	\label{table:optimal_strategies_before_liq}
\end{table}

\begin{figure}[H]
\centering
  \includegraphics[width=.75\linewidth]{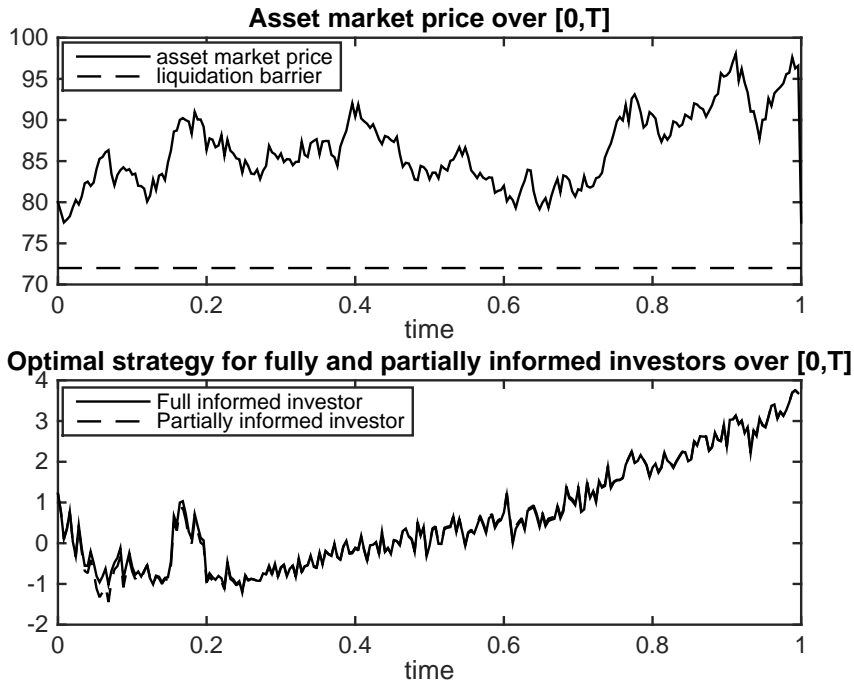}
  \caption{Approximated optimal strategy for fully and partially informed investors without liquidation}
\label{fig:optimal_strategy_power_No_Liq}
\end{figure}

Figure \ref{fig:optimal_strategy_power_No_Liq} illustrates the approximated optimal strategies for fully and partially investors respectively corresponding to a realized path of the asset price that does not induce liquidation.  In particular, the optimal trading strategies of the fully informed and partially informed investors appear almost identical.  We also observe a general tendancy for the optimal strategies to decrease the position in the stock  as its price moves toward the liquidation barrier and increase the position in the stock as the price moves away from the liquidation barrier.  However, as the time to the end of the investment horizon shortens and the probability of liquidation appears less likely the overall trend to  increase the position in the stock, toward the level of the Merton strategy, dominates.

\subsection{Optimal expected utility}
In this subsection we implement the Monte-Carlo method to find the optimal expected power and log utilities.
 In case of uninformed investors, since the
"optimal" strategy is simply the Merton strategy, we may approximate the wealth process $X^{(0)}$ directly using the Euler scheme.  For $0\leq m \leq M$ and  $1\leq n \leq N$, we denote by $X^{(0),n}_{t_{m}}$ the realized wealth
for uninformed investors at time $t_{m}$ in the nth simulation.  The expected utility $\mathbb{E}[U(X^{(0)})]$ is approximated by the sample mean
$\bar{V}^{(0)} = \frac{1}{N}\sum_{1\leq n \leq N}U(X^{(0),n}_{t_{M}})$. The standard error of the sample mean is
$$
  SE^{(0)}=\sqrt{\frac{1}{(N-1)N}\sum_{1\leq n \leq N}\left(U(X^{(0),n}_{t_{M}})-\bar{V}^{(0)}\right)^2}.$$
The relative standard error of the sample mean is
$ RSE^{(0)}= (SE^{(0)})/(|\bar{V}^{(0)}|)$. The $95\%$ confidence interval estimate of the sample mean is $[\bar{V}^{(0)}-1.96*SE^{(0)}, \bar{V}^{(0)}+1.96*SE^{(0)}]$.  This
simulation scheme also applies to the log utility for fully and partially informed investors.

However, in case of the power utility for fully and partially informed
investors we cannot approximate the wealth process directly since the optimal strategies are not explicitly determined. Although we can first approximate
the optimal strategies by solving the related BSDE, this would increase the size of simulation error. Instead we simulate the likelihood process $L$ in
(\ref{eq:likelihood_process_fully_informed}) and $\bar{L}$ in (\ref{eq:likelihood_process_partial_c}) since the optimal expected power utilities are
functionals of $L_{T}$ and $\bar{L}_{T}$ given by (\ref{eq:optimal_utility_power_fully_informed}) and (\ref{eq:optimal_utility_power_partially_informed})
respectively. For instance, in case of power utility for fully informed investors, we denote the discretized realization of $L_{t}$ in nth simulation by $L^{n}_{t_{m}}$ for
$0\leq m\leq M$ and $1\leq n \leq N$. The expectation $\mathbb{E}[(L_{T})^{\frac{p}{p-1}}]$ is estimated by the sample mean $\bar{\xi} = \frac{1}{N}
\sum_{1\leq n \leq N}(L^{n}_{t_{M}})^{\frac{p}{p-1}}$. The standard error of the sample mean is
$$SE^{(2)}=\sqrt{\frac{1}{(N-1)N}\sum_{1\leq n \leq N}
\left((L^{n}_{t_{M}})^{\frac{p}{p-1}}-\bar{\xi}\right)^2}.$$
The relative standard error of the sample mean is
$RSE^{(2)}=(SE^{(2)})/(|\bar{\xi}|)$. The $95\%$
confidence interval estimate of the sample mean is $[\bar{\xi}-1.96*SE^{(2)}, \bar{\xi}+1.96*SE^{(2)}]$.  By (\ref{eq:optimal_utility_power_fully_informed}) the optimal expected utility for fully informed investors is estimated by $\bar{V}^{(2)} = \frac{x_{0}^{p}}{p}(\bar{\xi})^{1-p}$. The $95\%$
confidence interval estimate of optimal expected utility is $[ \frac{x_{0}^{p}}{p}\left(\bar{\xi}-1.96*SE^{(2)}\right)^{1-p},  \frac{x_{0}^{p}}{p}\left(\bar{\xi}+1.96*SE^{(2)}\right)^{1-p}]$. A similar scheme can be applied to the case of power utility for partially informed investors.

 We present the numerical results on the optimal expected utilities for the three types of investors  in the Table 2 and Table 3 for power and log utilities respectively.  As should be expected there exists certain gaps among the optimal expected utilities of different types of investors. We may interpret those gaps as the value of information asymmetry.  The results are more pronounced in the case of power utility than in the case of power utility.  Nevertheless, in both cases there are statistically significant differences in optimal expected wealth given that the confidence intervals do not overlap. In the power utility case  the optimal strategy of the partially informed investor is very close to that of the fully informed investor.  However, the inability to fully capture the potential gains from trading against liquidation, owning to the need to estimate the drift and the tracking error, leads to a significantly lower optimal expected utility.

\begin{table}[H] \label{tab2}
	\centering
	\begin{tabular}{|c|c|c|c|}		
		\hline
		\multirow{3}{*}{Expected utilities} &
		\multicolumn{3}{c|}{Numerical evaluation  }\\
		\cline{2-4}
		& 	\multirow{2}{*}{Sample mean} & Relative & 95\% estimated  \\
		&  & standard error &  confidence interval\\
		\hline
		Fully informed  & 48.9602   &     0.0883    &     [44.5223, 53.0279] \\
		\hline
		Partially informed  & 31.3099   &   0.0172 &    [30.7767, 31.8342]\\
		\hline
		Uninformed  & 18.9228  &     0.0012   &      [18.8796, 18.9661]\\
		\hline
	\end{tabular}
	\caption{Numerical evaluation of optimal power utilities for three types of investors}
\end{table}
\begin{table}[H] \label{tab3}
	\centering
	\begin{tabular}{|c|c|c|c|}
		\hline
		\multirow{3}{*}{Expected utilities} &
		\multicolumn{3}{c|}{Numerical evaluation  }\\
		\cline{2-4}
		& 	\multirow{2}{*}{Sample Mean} & Relative  & 95\% estimated  \\
		&   & standard error &  confidence interval\\
		\hline
		Fully informed  & 4.8282  &  0.0073   &       [4.8219, 4.8346]\\
		\hline
		Partially informed  &4.7579  &    0.0080  &  [4.7520, 4.7638]\\
		\hline
		Uninformed  & 4.3665   & 0.0005 &   [4.3621, 4.3709] \\
		\hline
	\end{tabular}
	\caption{Numerical evaluation of optimal log utilities for three types of investors}
\end{table}

\section{Conclusion}\label{sect:conclusion}
In this paper, we characterize the market impact of liquidation by a function of certain form. We consider the portfolio optimization problem for three types of investors with different level of information about the liquidation trigger mechanism and the market impact. In case of logarithmic utility, we find the closed-form optimal strategy for all three types of investors.  In the case of power utility it is not as straightforward to find the closed-form optimal strategy for the partially informed investors. Finally we present some numerical results using Monte-Carlo simulation method.  These results indicate that there is significant value, in terms of optimal expected utility, of increased information about the opportunity to trade optimally against an investor who may need to liquidate a large position in a disorderly fashion.

There are several possible directions for improving the model. We can use more realistic models of market impact or the barrier that may depend on market, regulatory, or macro-economic variables. For partial insiders, the occupation time below the liquidation threshold is a random variable rather than a known constant as in case of full insiders. We plan to incorporate permanent price impact into the liquidation impact function generalizing the temporary price impact function. We shall explore the effect of different liquidation impact functions on the optimal trading strategies and utility of terminal wealth of uninformed, partially informed, and completely informed investors. Since certain market participants possess different market information it is natural to discuss the value of information in terms of portfolio utility. The results of future research can also inform financial and operational risk management processes and regulations for certain agents and trading activities including short-selling prohibitions, buying constraints, or derivatives market participation.

\section*{Acknowledgement}
The authors thank the Institut de finance structur\'{e}e et des instruments d\'{e}riv\'{e}s de Montr\'{e}al (IFSID) for funding this research. Caroline Hillairet acknowledges also support from Investissements d'Avenir (ANR-11-IDEX-0003/Labex Ecodec/ANR-11-LABX-0047)

\bibliographystyle{abbrvnat}
\bibliography{mybib}

\appendix		
\section{Appendix}
In the appendix we provide technical lemmas and proofs which allow us to easily justify the main results.

\begin{lemma}\label{lemma:optimal_utility_log_term_one}
\begin{align*}
& \mathbb{E}[1_{\{\tau>T\}}\ln (\hat  X^{(2,b)}_{T})|\mathcal G_0^{(2)}]=\left( \ln(X_{0})+\frac{1}{2}(\mu-\frac{\mu^{2}}{\sigma^{2}})T  \right)  \\
& \times  \left\{ \mathcal{N}\left(\frac{-\frac{\ln \alpha}{\sigma}+(\frac{\mu}{\sigma}-\frac{1}{2}\sigma) T}{\sqrt{T}}\right)-\exp\left(\frac{2\mu\ln\alpha}{\sigma^{2}}-\ln\alpha\right)\mathcal{N}\left(\frac{\frac{\ln \alpha}{\sigma}+(\frac{\mu}{\sigma}-\frac{1}{2}\sigma) T}{\sqrt{T}}\right) \right\} \\
&+ \int_{\frac{\ln\alpha}{\sigma}}^{0}\int_{y}^{\infty} \frac{2\mu x(x-2y)}{\sigma\sqrt{2\pi T^{3}}}\exp\left\{(\frac{\mu}{\sigma}-\frac{1}{2}\sigma) x-\frac{1}{2}(\frac{\mu}{\sigma}-\frac{1}{2}\sigma)^{2}T-\frac{1}{2T}(2y-x)^{2}\right \}dxdy
\end{align*}
where $\mathcal{N}(x)=\frac{1}{\sqrt{2\pi}}\int_{-\infty}^{x}e^{\frac{u^{2}}{2}}du$ is the cumulative distribution function of a standard normal random variable.
\end{lemma}

\begin{proof}
By (\ref{eq:risky_asset_fundamental_value}) we have
\begin{equation}
S_{t}=S_{0}\exp\left\{\sigma\left((\frac{\mu}{\sigma}-\frac{1}{2}\sigma )t+W_{t}\right)\right\}.
\label{eq:asset_price_SDE_solution}
\end{equation}
Define $B_{ t}=(\frac{\mu}{\sigma}-\frac{1}{2}\sigma)t+ W_{t}$ and
$\tilde B_{t}=\inf\{  B_v| 0\leq v\leq  t \}$. Recalling the definition of $\tau$ in (\ref{eq:liquidation_time_definition}) we find
\begin{equation}
1_{\{\tau>T\}}=1_{\{\tilde B_{T}>\frac{\ln \alpha}{\sigma} \}}.
\label{eq:liquidation_time_equivalent_definition}
\end{equation}
Let $\kappa=\frac{\mu}{\sigma}-\frac{1}{2}\sigma$. From \citet{jeanblanc2009mathematical} we know
\begin{equation}
\mathbb{P}(B_{T}\in dx,\tilde B_{T}\in dy)=1_{\{x>y\}}1_{\{y<0\}}\frac{2(x-2y)}{\sqrt{2\pi T^3}}\exp\{\kappa x-\frac{1}{2}\kappa^2T-\frac{1}{2T}(2y-x)^2\}dxdy.
\label{eq:Brownian_minimum_double_density}
\end{equation}
and
\begin{equation}
\mathbb{P}(\tau>T)=\mathcal{N}\left(\frac{-\frac{\ln \alpha}{\sigma}+(\frac{\mu}{\sigma}-\frac{1}{2}\sigma) T}{\sqrt{T}}\right)-\exp\left(\frac{2\mu\ln(\alpha)}{\sigma^2}-\ln(\alpha)\right)\mathcal{N}\left(\frac{\frac{\ln \alpha}{\sigma}+(\frac{\mu}{\sigma}-\frac{1}{2}\sigma) T}{\sqrt{T}}\right).
\label{eq:liquidation_time_greater_than_T_prob}
\end{equation}
On the other hand, by (\ref{eq:wealth_fully_informed_log_before_liquidation_SDE})  we know
\begin{align}
\hat X^{(2,b)}_{T} & = X_{0}\exp\left\{\left(\hat\pi^{(2,b)}\mu - \frac{1}{2}(\hat \pi^{(2,b)}\sigma)^{2}\right) T+\hat\pi^{(2,b)}\sigma W_{T}\right\}\nonumber\\
 &=X_{0}\exp\left\{ \frac{\mu}{\sigma}\left( ( \frac{\mu}{\sigma}-\frac{1}{2}\sigma)T+W_{T} \right)+\frac{1}{2}\mu(1-\frac{\mu}{\sigma^{2}})T\right\}\nonumber\\
&=X_{0}\exp\left\{ \frac{\mu}{\sigma}B_{T}+\frac{1}{2}\mu(1-\frac{\mu}{\sigma^{2}})T\right\}. \label{eq:wealth_no_liquidation_terminal_value}
\end{align}
Using  (\ref{eq:liquidation_time_equivalent_definition}) and (\ref{eq:wealth_no_liquidation_terminal_value}) we compute
\begin{align}
& \mathbb{E}\left[1_{\{\tau>T\}}\ln\left(\hat X^{(2,b)}_{T}\right) |\mathcal G_0^{(2)}\right] \nonumber \\
=& \mathbb{E}\left[1_{\{\tau>T\}}\left\{ \ln(X_{0})+ \frac{\mu}{\sigma}B_{T}+\frac{1}{2}\mu(1-\frac{\mu}{\sigma^{2}})T  \right\} |\mathcal G_0^{(2)} \right]\nonumber \\
=& \mathbb{P}(\tau>T)\left\{ \ln(X_{0})+\frac{1}{2}(\mu-\frac{\mu^{2}}{\sigma^{2}})T \right\}+\mathbb{E}\left[1_{ \{ \hat B_{T}> \frac{\ln\alpha}{\sigma} \} } \frac{\mu}{\sigma} B_{T}  \right] \label{eq:log_utility_first_expectation_compuation}
\end{align}	
since $(\Theta, K)$ is independent of $\mathbb F$ and $X_0$ is $\mathcal G_0^{(2)}$-measurable.
Finally we apply (\ref{eq:Brownian_minimum_double_density}) and (\ref{eq:liquidation_time_greater_than_T_prob}) to (\ref{eq:log_utility_first_expectation_compuation}) to obtain the result.
\end{proof}

\begin{lemma}\label{lemma:optimal_utility_log_term_two}
\begin{align*}
& \mathbb{E}[1_{T\geq \tau}\ln(\hat  X^{(2,a)}_{T})|\mathcal G_0^{(2)}] \\
= & -\frac{\ln\alpha}{\sigma}\int_{0}^{T}\frac{1}{\sqrt{2\pi t^{3}}}\exp\left\{-\frac{1}{2t}\left(\frac{\ln\alpha}{\sigma}-(\frac{\mu}{\sigma}-\frac{1}{2}\sigma) t\right)^{2}\right\}h^{(2)}(t,\Theta,K)dt
\end{align*}
where
\begin{equation}
h^{(2)}(t,\theta,k):=\ln X_{0}+\frac{\mu\ln\alpha}{\sigma^{2}} +\frac{\mu}{2}t-\frac{\mu^{2}}{2\sigma^{2}}t +\int_{t}^{T}\frac{\left(\mu^{I}_v(t,\theta,k)\right)^{2}}{2\sigma^2}dv.
\label{eq:h_function_in_lemma_fully_informed}
\end{equation}
\end{lemma}

\begin{proof}
Let $t = \tau$ in (\ref{eq:asset_price_SDE_solution}) we have
\begin{equation}
S_{\tau}=S_{0}\exp\{(\mu-\frac{1}{2}\sigma^{2})\tau+\sigma W_{\tau}^{\mathbb{P}}\}.
\label{eq:asset_price_at_liquidation_expression}
\end{equation}
Using the fact $S_{\tau}=\alpha S_{0}$ we find
\begin{equation}
W_{\tau}^{\mathbb{P}}=\frac{1}{\sigma}\left\{ \ln \alpha-(\mu-\frac{1}{2}\sigma^2)\tau  \right\}.
\label{eq:brownian_at_liquidation_value}
\end{equation}
By (\ref{eq:wealth_fully_informed_log_before_liquidation_SDE})  and (\ref{eq:brownian_at_liquidation_value}) we compute
\begin{align}
\hat X^{(2,b)}_{\tau} & = X_{0}\exp\left\{(\hat{\pi}^{(2,b)}-\frac{1}{2}(\hat \pi^{(2,b)})^{2}\sigma^{2})\tau+\hat\pi^{b}\sigma W_{\tau}\right\} \nonumber \\
& = X_{0}\exp\left\{\frac{\mu^{2}}{2\sigma^{2}}t+\frac{\mu}{\sigma} W_{\tau}\right\} \nonumber\\
& = X_{0}\exp\left\{\frac{\mu\ln\alpha}{\sigma^{2}} +\frac{\mu}{2}\tau-\frac{\mu^{2}}{2\sigma^{2}}\tau\right\}. \label{eq:wealth_fully_informed_log_at_liquidation}
\end{align}	
Solving (\ref{eq:wealth_fully_informed_log_after_liquidation_SDE}) we obtain
\begin{align}
\hat X^{(2,a)}_{T} & =  \hat X^{(2,b)}_{\tau}\exp\{\int_{\tau}^{T}(\hat \pi^{2,a}_{v}(\Theta,K)\mu^{I}_{v}(\tau,\Theta,K)-\frac{1}{2}(\hat \pi^{2,a}_{v}(\Theta,K))^{2}\sigma^{2})dv+\int_{\tau}^{T}\hat\pi^{2,a}_{v}(\Theta,K)\sigma dW_{v}\} \nonumber\\
& = \hat X^{(2,b)}_\tau\exp\{\int_{\tau}^{T}\frac{(\mu^{I}_{v}(\tau, \Theta,K))^{2}}{2\sigma^{2}}dv+\int_{\tau}^{T}\frac{\mu^{I}_{v}(\tau, \Theta,K)}{\sigma} dW_{v}\}.
\label{eq:wealth_log_fully_informed_terminal_expression}
\end{align}
Using (\ref{eq:wealth_fully_informed_log_at_liquidation}) and (\ref{eq:wealth_log_fully_informed_terminal_expression}) we compute
\begin{align*}
& \mathbb{E}\left[1_{T\geq \tau}\ln \left(\hat X^{(2,a)}_{T} \right)| \mathcal G_0^{(2)} \right] \\
=&  \mathbb{E}\left[1_{T\geq \tau}\left(\ln\left(\hat X^{(2,b)}_{\tau}\right)+{\int_{\tau}^{T} \frac{\left(\mu_{v}^{I}(\tau, \Theta,K)\right)^{2}}{2\sigma^{2}}dv+\int_{\tau}^{T}\frac{\mu^{I}_{v}(\tau, \Theta,K)}{\sigma}dW_{v}}\right)| \mathcal G_0^{(2)}\right] \\
=& \mathbb{E}\left[\mathbb{E}\left[\left.1_{T\geq \tau}\left(\ln(\hat X^{(2,b)}_{\tau})+{\int_{\tau}^{T} \frac{\left(\mu_{v}^{I}(\tau, \Theta,K)\right)^{2}}{2\sigma^{2}}dv+\int_{\tau}^{T}\frac{\mu^{I}_{v}(\tau, \Theta,K)}{\sigma}dW_{v}}\right)\right|\sigma(\tau), \mathcal G_0^{(2)}\right] | \mathcal G_0^{(2)} \right] \\
=&  \mathbb{E}\left[1_{T\geq \tau}\left(\ln X_{0}+\frac{\mu\ln\alpha}{\sigma^{2}} +\frac{\mu}{2}\tau-\frac{\mu^{2}}{2\sigma^{2}}\tau +\int_{\tau}^{T}\frac{\left(\mu^{I}_{v}(t,\Theta,K)\right)^{2}}{2\sigma^{2}}dv\right) | \mathcal G_0^{(2)} \right] .
\end{align*}
Recall that $(\Theta,K)$ is independent to $\mathbb F$ and that from \citet[Sect. 3.3.1]{jeanblanc2009mathematical}  that the density of $\tau$ is
\begin{equation}
\mathbb{P}(\tau\in dt)=-\frac{\ln\alpha}{\sigma}\frac{1}{\sqrt{2\pi t^3}}\exp\left\{-\frac{1}{2t}\left(\frac{\ln\alpha}{\sigma}-(\frac{\mu}{\sigma}-\frac{1}{2}\sigma) t\right)^2\right\}dt.
\label{eq:tau_desntiy}
\end{equation}
Using (\ref{eq:tau_desntiy}),  and the definition of the function $h^{(2)}(t,\theta,k)$ in (\ref{eq:h_function_in_lemma_fully_informed}), we obtain the result.
\end{proof}

\begin{lemma}\label{lemma:optimal_utility_partially_log_term_two}
\begin{equation}
\mathbb{E}[1_{T\geq \tau}\ln(\hat {X}^{1,a}_{T})]
= -\frac{\ln\alpha}{\sigma}\int_{0}^{T}\frac{1}{\sqrt{2\pi t^{3}}}\exp\left\{-\frac{1}{2t}\left(\frac{\ln\alpha}{\sigma}-(\frac{\mu}{\sigma}-\frac{1}{2}\sigma) t\right)^2\right\}h^{(1)}(t)dt
\end{equation}
where
\begin{equation}
h^{(1)}(t):=\ln x_{0}+\frac{\mu\ln\alpha}{\sigma^{2}} +\frac{\mu}{2}t-\frac{\mu^{2}}{2\sigma^{2}}t +\int_{t}^{T}\frac{\left(\bar{\mu}^{M}_{v}\right)^{2}}{2\sigma^2}dv.
\label{eq:h_function_in_lemma_partially_informed}
\end{equation}
\end{lemma}

\begin{proof}
Similar to the proof of Lemma \ref{lemma:optimal_utility_log_term_two}, we find the terminal wealth $\hat{X}^{(1,a)}_{T}$ if liquidation occurs before $T$
\begin{equation}
\hat X^{(1,a)}_{T}
= x_{0}\exp\left\{\frac{\mu\ln\alpha}{\sigma^{2}} +\frac{\mu}{2}\tau-\frac{\mu^{2}}{2\sigma^{2}}\tau+\int_{\tau}^{T}\frac{(\bar\mu^{M}_{v})^{2}}{2\sigma^{2}}dv+\int_{\tau}^{T}\frac{\bar\mu^{M}_{v}}{\sigma} d\tilde{W}_{v}\right\}.
\label{eq:wealth_log_partially_informed_terminal_expression}
\end{equation}
We compute
\begin{align}
& \mathbb{E}\left[1_{T\geq \tau}\ln \left(\hat X^{(1,a)}_T\right)\right] \nonumber\\
=& \mathbb{E}\left[1_{T\geq \tau}\left\{\ln x_{0}+\frac{\mu\ln\alpha}{\sigma^{2}} +\frac{\mu}{2}\tau-\frac{\mu^{2}}{2\sigma^{2}}\tau+\int_{\tau}^{T}\frac{(\bar\mu^{M}_{v})^{2}}{2\sigma^{2}}dv+\int_{\tau}^{T}\frac{\bar\mu^{M}_{v}}{\sigma} d\tilde{W}_{v}\right\}\right] \nonumber\\
=& \mathbb{E}\left[\mathbb{E}\left[\left.1_{T\geq \tau}\left\{\ln x_{0}+\frac{\mu\ln\alpha}{\sigma^{2}} +\frac{\mu}{2}\tau-\frac{\mu^{2}}{2\sigma^{2}}\tau+\int_{\tau}^{T}\frac{(\bar\mu^{M}_{v})^{2}}{2\sigma^{2}}dv+\int_{\tau}^{T}\frac{\bar\mu^{M}_{v}}{\sigma} d\tilde{W}_{v}\right\}\right|\sigma(\tau)\right]\right]\nonumber\\
=& \mathbb{E}\left[1_{T\geq \tau}\left\{\ln x_{0}+\frac{\mu\ln\alpha}{\sigma^{2}} +\frac{\mu}{2}\tau-\frac{\mu^{2}}{2\sigma^{2}}\tau+\int_{\tau}^{T}\frac{(\bar\mu^{M}_{v})^{2}}{2\sigma^{2}}dv\right\}\right]
\end{align}
Using the density of $\tau$ given in (\ref{eq:tau_desntiy}) and the definition of the function $h^{(1)}(t)$ in (\ref{eq:h_function_in_lemma_partially_informed}), we obtain the result.
\end{proof}

\begin{lemma}\label{lemma:optimal_utility_power_term_one_uninformed}
\begin{align}
& \frac{1}{p}\mathbb{E}[1_{\{\tau>T\}}(\hat X^{(0,b)}_{T})^p] =
\frac{x_{0}^{p}}{p}\exp\left(\frac{p\mu^{2}T}{2(1-p)\sigma^{2}} \right) \nonumber\\
&\times\left\{\mathcal{N}\left(\frac{-\frac{\ln \alpha}{\sigma}+(\frac{\mu}{(1-p)\sigma}-\frac{\sigma}{2})T}{\sqrt{T}}\right)-\exp\left(\frac{2\mu\ln \alpha}{(1-p)\sigma^{2}}-\ln \alpha\right)\mathcal{N}\left(\frac{\frac{\ln \alpha}{\sigma}+(\frac{\mu}{(1-p)\sigma}-\frac{\sigma}{2})T}{\sqrt{T}}\right)\right\}.
\end{align}
\end{lemma}

\begin{proof}
The proof is basically the same as that of Lemma \ref{lemma:optimal_utility_log_term_one} except using the power utility function instead of log utility function. We compute
\begin{align}
& \frac{1}{p}\mathbb{E}[1_{\{\tau>T\}}(\hat X^{(0,b)}_{T})^p] \nonumber\\
=&\frac{x_{0}^{p}}{p}\exp\left(\frac{p\mu^{2}T}{2(1-p)\sigma^{2}} \right)  \nonumber\\
&\times\int_{\frac{\ln\alpha}{\sigma}}^{0}\int_{y}^{\infty} \frac{2(x-2y)}{\sqrt{2\pi T^{3}}}\exp\left\{(\frac{\mu}{(1-p)\sigma}-\frac{\sigma}{2})x-\frac{1}{2}(\frac{\mu}{(1-p)\sigma}-\frac{\sigma}{2})^{2}T-\frac{1}{2T}(2y-x)^2\right\}dxdy.
\label{eq:power_utility_first_expectation_compuation}
\end{align}	
Define $ C_t=(\frac{\mu}{(1-p)\sigma}-\frac{\sigma}{2})t+W_t$ and $\tilde C_t=\inf\{ C_s| 0\leq s\leq t \}$. Note that the integral in (\ref{eq:power_utility_first_expectation_compuation}) is  equal to $\mathbb{P}(\tilde C_{T}>\frac{\ln \alpha}{\sigma})$.
By \citet[Sect. 3.2.2]{jeanblanc2009mathematical} we know
\begin{align}
&\mathbb{P}(\tilde C_{T}>\frac{\ln \alpha}{\sigma})\nonumber\\
=&\mathcal{N}\left(\frac{-\frac{\ln \alpha}{\sigma}+(\frac{\mu}{(1-p)\sigma}-\frac{\sigma}{2})T}{\sqrt{T}}\right)-\exp\left(\frac{2\mu\ln \alpha}{(1-p)\sigma^2}-\ln \alpha\right)\mathcal{N}\left(\frac{\frac{\ln \alpha}{\sigma}+(\frac{\mu}{(1-p)\sigma}-\frac{\sigma}{2})T}{\sqrt{T}}\right).
\label{eq:C_T_cumulative_probability}
\end{align}
Substituting (\ref{eq:C_T_cumulative_probability}) into (\ref{eq:power_utility_first_expectation_compuation}) we obtain the result.
\end{proof}

\begin{lemma}\label{lemma:optimal_utility_power_term_two_uninformed}
\begin{align*}
& \frac{1}{p}\mathbb{E}[1_{T\geq \tau}(\hat X^{(0,a)}_{T})^{p}]\\
=& -\frac{\ln\alpha}{\sigma}\int_{0}^{1}\int_{0}^{\infty}\int_{0}^{T}\frac{1}{\sqrt{2\pi t^{3}}}\exp\left\{-\frac{1}{2t}\left(\frac{\ln\alpha}{\sigma}-(\frac{\mu}{\sigma}-\frac{1}{2}\sigma) t\right)^{2}\right\}l^{(0)}(t,\theta,k)\varphi(\theta,k)dtd\theta dk
\end{align*}
where
\begin{equation}
l^{(0)}(t,\theta,k)=\frac{x_{0}^{p}}{p}\exp\left\{\frac{\mu\ln\alpha}{(1-p)\sigma^{2}} +\frac{1}{2}\frac{\mu}{(1-p)}t-\frac{1}{2}\frac{\mu^{2}}{(1-p)^{2}\sigma^2}t+\int_{t}^{T}\left( \frac{p\mu\mu^{I}_{v}(t,\theta,k)}{(1-p)\sigma^{2}} \right)dv\right\}.
\label{eq:l_function_in_lemma_uninformed}
\end{equation}
\end{lemma}

\begin{proof}
Similar to the proof of Lemma \ref{lemma:optimal_utility_log_term_two} we find the wealth value at liquidation time $\tau$ as follows
\begin{equation*}
\hat X^{(0,b)}_\tau=x_0\exp\left\{\frac{\mu\ln\alpha}{(1-p)\sigma^2} +\frac{1}{2}\frac{\mu}{(1-p)}\tau-\frac{1}{2}\frac{\mu^2}{(1-p)^2\sigma^2}\tau\right\}.
\end{equation*}
Then the terminal wealth $\hat X^{(0,a)}_{T}$ if liquidation occurs before $T$ is
\begin{equation*}
\hat X^{(0,a)}_T=\hat X^{(0,b)}_\tau \exp\left\{\int_{\tau}^{T}\left( \frac{\mu\mu^I_v(\tau,\Theta,K)}{(1-p)\sigma^2}-\frac{\mu^2}{2(1-p)^2\sigma^2}\right)dv+\int_{\tau}^{T}\frac{ \mu}{(1-p)\sigma}dW_{v}\right\}.
\end{equation*}
We compute
\begin{align*}
& \mathbb{E}[1_{T\geq \tau}U (\hat X^{(0,a)}_T)] \\
=&\frac{1}{p}\mathbb{E}\left[1_{\{T\geq \tau\}}(\hat X^{b}_{\tau})^{p}\exp\left\{\int_{\tau}^{T}\left( \frac{p\mu\mu^{I}_{v}(\tau,\Theta,K)}{(1-p)\sigma^{2}} -\frac{p\mu^{2}}{2(1-p)^{2}\sigma^{2}}\right)dv+\int_{\tau}^{T}\frac{ p\mu}{(1-p)\sigma}dW_{v}\right\}\right] \\
=& \frac{1}{p}\mathbb{E}\left[\mathbb{E}\left[\left.1_{\{T\geq \tau\}}(\hat X^{b}_{\tau})^{p}\exp\left\{\int_{\tau}^{T}\left( \frac{p\mu\mu^{I}_{v}(\tau,\Theta,K)}{(1-p)\sigma^{2}} -\frac{p\mu^{2}}{2(1-p)^{2}\sigma^{2}}\right)dv+\int_{\tau}^{T}\frac{ p\mu}{(1-p)\sigma}dW_{v}\right\}\right|\sigma(\tau)\right]\right] \nonumber\\
= &  \frac{1}{p}\mathbb{E}\left[1_{\{T\geq \tau\}}(\hat X^{b}_{\tau})^{p}\exp\left\{\int_{\tau}^{T}\left( \frac{p\mu\mu^{I}_{v}(\tau, \Theta,K)}{(1-p)\sigma^{2}} -\frac{p\mu^{2}}{2(1-p)\sigma^{2}}\right)dv\right\}\right]
\end{align*}
Using the density of $\tau$ given in (\ref{eq:tau_desntiy}) and the definition of the function $l^{(0)}(t,\theta,k)$ in (\ref{eq:l_function_in_lemma_uninformed}), we obtain the result.
\end{proof}

\begin{lemma}\label{lemma:optimal_utility_log_term_two_uninformed}
\begin{align*}
&\mathbb{E}[1_{T\geq \tau}\ln(\hat X^{(0,a)}_T)] \\
=& -\frac{\ln\alpha}{\sigma}\int_{0}^{\infty}\int_{0}^{\infty}\int_0^T\frac{1}{\sqrt{2\pi t^{3}}}\exp\left\{-\frac{1}{2t}\left(\frac{\ln\alpha}{\sigma}-(\frac{\mu}{\sigma}-\frac{1}{2}\sigma) t\right)^2\right\}h^{(0)}(t,\theta,k)\varphi(\theta,k)dtd\theta dk
\end{align*}
where
\begin{equation}
h^{(0)}(t,\theta,k):=\ln x_{0}+\frac{\mu\ln\alpha}{\sigma^{2}} +\frac{\mu}{2}t-\frac{\mu^{2}}{2\sigma^{2}}t+\int_{t}^{T}\left( \frac{2\mu\mu^{I}_{v}(t,\theta,k)- \mu^2}{2\sigma^{2}} \right)dv.
\label{eq:h_function_in_lemma_uninformed}
\end{equation}
\end{lemma}
\begin{proof}
Similar to the proof of Lemma \ref{lemma:optimal_utility_log_term_two}, we find the terminal wealth $\hat{X}^{(0,a)}_{T}$ if liquidation occurs before $T$
\begin{equation}
\hat X^{(0,a)}_T=x_{0} \exp\left\{\frac{\mu\ln\alpha}{\sigma^{2}} +\frac{\mu}{2}\tau-\frac{\mu^{2}}{2\sigma^{2}}\tau+\int_{\tau}^{T}\left( \frac{\mu\mu^{I}_{v}(\tau,\Theta,K)}{\sigma^{2}}-\frac{\mu^{2}}{2\sigma^{2}} \right)dt+\int_{\tau}^{T}\frac{ \mu}{\sigma}dW_{t}\right\}.
\end{equation}
We compute
\begin{align*}
& \mathbb{E}\left[1_{T\geq \tau}\ln(\hat X^{(0,a)}_T)\right] \\
=& \mathbb{E}\left[1_{T\geq \tau}\left\{\ln x_{0}+\frac{\mu\ln\alpha}{\sigma^{2}} +\frac{\mu}{2}\tau-\frac{\mu^{2}}{2\sigma^{2}}\tau+\int_{\tau}^{T}\left( \frac{\mu\mu^{I}_{v}(\tau, \Theta,K)}{\sigma^{2}}-\frac{\mu^{2}}{2\sigma^{2}} \right)dt+\int_{\tau}^{T}\frac{ \mu}{\sigma}dW_{t}\right\}\right] \\
=& \mathbb{E}\left[\mathbb{E}\left[\left.1_{T\geq \tau}\left\{\ln x_{0}+\frac{\mu\ln\alpha}{\sigma^{2}} +\frac{\mu}{2}\tau-\frac{\mu^{2}}{2\sigma^{2}}\tau+\int_{\tau}^{T}\left( \frac{\mu\mu^{I}_{v}(\tau,\Theta,K)}{\sigma^{2}}-\frac{\mu^{2}}{2\sigma^{2}} \right)dt+\int_{\tau}^{T}\frac{ \mu}{\sigma}dW_{t}\right\}\right|\sigma(\tau)\right] \right]\\
=& \mathbb{E}\left[1_{T\geq \tau}\left\{\ln x_{0}+\frac{\mu\ln\alpha}{\sigma^{2}} +\frac{\mu}{2}\tau-\frac{\mu^{2}}{2\sigma^{2}}\tau+\int_{\tau}^{T}\left( \frac{\mu\mu^{I}_{v}(\tau,\Theta,K)}{\sigma^{2}}-\frac{\mu^{2}}{2\sigma^{2}} \right)dt\right\}\right]
\end{align*}
Using the density of $\tau$ given in (\ref{eq:tau_desntiy}) and the definition of the function $h^{(0)}(t,\theta)$ in (\ref{eq:h_function_in_lemma_uninformed}), we obtain the result.
\end{proof}

\end{document}